\documentclass[%
reprint,
amsmath,amssymb,
aps,
]{revtex4-2}

\usepackage{graphicx}
\usepackage{dcolumn}
\usepackage{bm}

\usepackage{amsmath,amsthm,amssymb,amsxtra, graphicx}
\usepackage{color}

\usepackage{tikz}
\usetikzlibrary{positioning} 
\usetikzlibrary{arrows,decorations.pathmorphing,shapes} 
\usepackage{tikz-cd}

\usepackage[utf8]{inputenc}
\usepackage[english]{babel}
\usepackage{amsmath,amstext,amsbsy,amsfonts,eucal,latexsym,amssymb}
\usepackage{indentfirst}
\usepackage{graphicx}
\usepackage{setspace}
\usepackage{float}
\usepackage{placeins}
\usepackage{xparse}
\usepackage{tikz}
\usepackage{braket}
\usepackage{longtable}
\usepackage{booktabs}

\usepackage{tabularx,ragged2e,booktabs}

\newtheorem{theorem}{Theorem}[section]
\newtheorem{proposition}{Proposition}
\newtheorem{corollary}{Corollary}

\newtheorem{definition}{Definition}

\newcommand{\tes}{\mathcal{T}}
\newcommand{\M}{\mathbb{M}}

\newtheorem{remark}[theorem]{{\bf Remark}}

\begin{document}
	
	
	\title{Floquet Codes from Derived Semi-Regular Hyperbolic Tessellations \\ on Orientable and Non-Orientable Surfaces}
	
	\author{Douglas F. Copatti}
	\affiliation{Department of Mathematics, Instituto Federal do Paraná, Campus Pitanga, Rua José de Alencar, 1.080 – Vila Planalto - Pitanga - PR, 85.200-000, Brazil}
\email{douglascopatti@gmail.com}

	\author{Giuliano G. La Guardia}%
	\email{gguardia@uepg.br}
	\affiliation{Department of Mathematics and Statistics, State University of Ponta Grossa, UEPG, Ponta Grossa - PR, 84030-900, Brazil.}
	
	
	\author{Waldir S. Soares}
	\affiliation{Department of Mathematics, Campus de Pato Branco, UTFPR, Universidade Técnol\'ogica Federal do Paran\'a, Via do Conhecimento, s/n-KM 01-Fraron, Pato Branco - PR, 85503-390, Brazil}
	\email{waldirjunior@utfpr.edu.br}
	
	\author{Edson D. Carvalho}
	\affiliation{Department of Mathematics, UNESP, Ilha Solteira - SP, 15385-007, Brazil}
	\email{edson.donizete@unesp.br}
	

\author{Eduardo B. Silva}
\affiliation{Department of Mathematics, Maring\'a State University - UEM, Av. Colombo 5790, Maring\'a - PR, 87020-900, Brazil}
	\email{ebsilva@uem.br}
	
	\date{\today}
	
	\begin{abstract}
In this paper, we construct several new quantum Floquet codes on compact, orientable, as well as non-orientable surfaces. In order to obtain such codes, we identify these surfaces with hyperbolic polygons and examine hyperbolic semi-regular tessellations on such surfaces. The method of construction presented here generalizes similar constructions concerning hyperbolic Floquet codes on connected and compact surfaces with genus $g \geq 2$. A performance analysis and an investigation of the asymptotic behavior of these codes are also presented.
	\end{abstract}
	
\maketitle
	

\section{Introduction}

To create a practical and usable quantum computer, it is necessary to build a fault-tolerant circuit to implement a specific code; it is important to choose a code that can work with simple check operators. Thus, it is of fundamental importance to build a circuit that is fault tolerant during the implementation of a quantum code. The use of low-weight verification codes is an effective way to achieve this goal. Another approach is to use codes that generalize stabilizer codes, which are referred to as subsystem codes \cite{Poulin1}. These codes enable the measurement of high-weight stabilizers using the product of low-weight verification operators' measurement results.

The family of Floquet codes \cite{Hastings2021} emerged as a generalization of subsystem codes \cite{8}. Floquet codes, like subsystem codes, require measuring low-weight check operators, but they are different because the code's logical operators can change over time. Error correction works much better in Floquet codes because they can change their logical operators as needed, which helps to utilize resources more effectively. This characteristic can lead to improved fault tolerance and faster recovery from errors, making Floquet codes particularly advantageous in quantum computing applications where maintaining coherence is critical. Maintaining coherence is critical for the successful implementation of quantum algorithms. Additionally, the changing characteristics of Floquet codes make it easier to use them with current quantum systems, which could result in stronger and more scalable quantum technologies. The time dependence of logical operators in Floquet codes introduces a unique dynamic that can enhance error correction by adapting to varying noise environments. This adaptability allows for more efficient encoding of quantum information, ultimately improving the performance and reliability of quantum algorithms in practical applications.

These features help create a stronger way to fix errors and allow quantum algorithms to work in real-life situations where noise is unpredictable. As quantum technology improves, using these time-changing logical operators can be key to creating larger quantum computing systems that can handle complicated tasks. Floquet codes can offer strong error correction methods that protect quantum information from the unavoidable decay and interference found in real-life settings. By using the special features of low-weight check operators, researchers can make quantum systems better at handling errors, which is crucial for practical uses in quantum communication and networks where being reliable is essential. Reliability is paramount, especially as quantum technologies move closer to real-world deployment. Continued research into Floquet codes will not only improve error correction techniques but also pave the way for more resilient quantum protocols that can withstand the challenges posed by environmental noise and operational imperfections. Consider challenges that may arise in implementing Floquet codes in real-world quantum systems and propose solutions. Any tiling suitable for defining a color code can yield Floquet codes, as demonstrated in \cite{Vuillot2021}. Additionally, research indicated that Floquet codes created from color code tilings on closed hyperbolic surfaces have a steady encoding rate and a distance that grows logarithmically; see \cite{Breuckmann2024}. These codes are referred to as hyperbolic Floquet codes. Hyperbolic Floquet codes show enormous potential for reliable quantum computing by using the special shapes of hyperbolic surfaces. Their robust performance in encoding quantum information makes them a significant focus of ongoing research in the field of quantum error correction. 

Researchers have employed hyperbolic geometry to develop novel topological codes on surfaces with a genus of $g \geq 2$. The works \cite{clarice}, \cite{albuquerque2009}, \cite{brandaniselfdual}, and \cite{Breuckmann2016} deal with surface codes. For color codes on closed surfaces with a genus of $g \geq 2$, one of the earliest works is \cite{waldir2018}. However, we have noticed that all the cited works' code constructions are based on regular tessellations.

In contrast to regular tessellations, which only consider regular and congruent polygons, semi-regular tessellations use more than one type of polygon. A tessellation is semi-regular if it consists of multiple types of regular polygons, denoted as $P_1, P_2, ...P_m$, such that, at each vertex, starting from polygon $P_1$, for instance, and proceeding clockwise, the sequence of polygons consistently appears in the identical order $P_1, P_2, ...P_m$. If the polygon $P_i$ has $n_i$ edges, $i=1,...,m$, then we denote the semi-regular tessellation by $[n_1, n_2, \ldots,n_m]$.

We draw attention to the fact that, for non-orientable surfaces, there is no work available in the literature that guarantees the existence of families of semi-regular tessellations. Then, we present semi-regular tessellation generation techniques, which are very interesting mathematical results by themselves. From these, we construct several codes whose parameters are new. We present tables for several surface genera, detailing the semi-regular tessellations considered on that surface, along with the parameters of the codes produced by each, including those created from regular tessellations.

The paper is organized as follows. Section~\ref{sec:semireg} presents a review of the concepts of regular and semi-regular tessellations. We present the general concepts and metric results for semi-regular tessellations on non-orientable surfaces. Section~\ref{reviewFloquet} brings a brief review of quantum Floquet codes. In Section~\ref{NewFloquet}, we present the contributions of the paper: constructions of new Floquet codes. We consider semi-regular tessellations derived from a regular $\{p, q\}$ tessellation, obtained in \cite{Douglas1}, and we also obtain new semi-regular tessellations using the techniques of \cite{antonio}. Then, we use the tools of hyperbolic geometry to construct families of Floquet codes on compact, orientable, and non-orientable surfaces of genus $g\geq 2$ from semi-regular tesselations. Section~\ref{CAB} is devoted to analyzing the asymptotic behavior of the new Floquet codes, and in Section~\ref{FR}, we summarize the main results of the paper. 

\section{Uniform tilings} \label{sec:semireg}

In this section, basic terminology and concepts concerning tilings are provided, mainly taken from \cite{artigo-giovana} and \cite{Douglas1}.

Let $\M$ be a connected Riemann surface without boundary. An isometry of $X$ is a distance-preserving bijection from $\M$ to $\M$. A tessellation, or tiling, is a countable family $\tes = (A_1, \hdots, A_n, \hdots)$ of connected closed subsets of $\M$, which are called tiles or faces. The surface is the union of all sets $A_i$, and the points where two faces meet are either empty, along edges, or at vertices. There is no intersection between the interiors of two separate sets, $A_i$ and $A_j$. 

The tiling's edges are the arcs of the elements $A_i$. The vertices in a tiling are the arcs' intersections and final points. A tile's topological boundary divides into a sequence of edges separated by vertices. Each edge connects two vertices and separates two tiles. We refer to the number of incident edges on a vertex $v$ as the vertex's degree or valence. A tiling face is the topological interior of a tile of $\tes$. In particular, we will be working with tilings where the tiles are regular polygons, all with the same edge size.

A symmetry of a tiling is an isometry of $\M$ that leaves the tiling globally invariant. A tiling is considered uniform when its group of symmetries $S$ is transitive on the set of vertices, i.e., there exists a symmetry $\gamma$ such that $\gamma(v_1) = v_2$ for each pair of vertices $v_1$ and $v_2$, and its tiles are regular polygons. We classify a tiling as regular if it is uniform and contains only one type of tile. On the other hand, a semi-regular tessellation of the plane is one in which each polygon is regular and the type of vertex is identical for each vertex. Both regular and semi-regular tessellations are uniform tessellations, according to the definition, and we will use this terminology in our current work.

A tiling $\tes$ is locally finite if all compact sets intersect a finite number of tiles of the tiling. Every compact subset in this case has a finite number of edges and vertices of $\tes$, and the degree of a vertex is finite. We can represent each tessellation as follows: $\tes = (V, E, F)$, where $V$, $E$, and $F$ are the sets of vertices, edges, and faces, respectively. We draw attention to the fact that in our current work, we are only considering locally finite tessellations, which means that each vertex has a finite number of faces. 

Two faces are adjacent if their boundaries share an edge or vertex. The size of a face is determined by the number of vertices along its boundary. This work focuses only on faces of finite size. Consider the set of all faces incident to a vertex $v$, and let $[a_1, \hdots, a_{q_v}]$ be a vector containing the faces' sizes, listed in cyclic order around $v$. This vector is a type vector or a type of vertex for $v$, and $q_v$ is the valence of the vertex.

Given an integer $q_v \geq 3$ and a vector of natural numbers $a = [a_1, a_2, \hdots, a_{q_v}]$, a necessary condition for $a$ to be a type of vertex is that they satisfy the condition 
\begin{equation}
	\label{cond}	
	\sum_{i =1}^{q_v}\frac{(a_i-2)\pi}{a_i}>2\pi \,.
\end{equation}

A regular tessellation whose polygons constituting the faces are $p$-gons and whose valence of each vertex is $q$, is denoted by $\{p, q\}$.

The Killing-Hopf theorem deals, in its original form, with Riemannian manifolds. In our case, we consider it a simple restriction. More details of this can be found at \cite{stillwell1995geometry}.

\begin{theorem}	(Killing-Hopf General, \cite{stillwell1995geometry}) 
	Any surface of constant curvature that is complete and connected is a quotient of an 
	Euclidean, hyperbolic or spherical space, by a group of isometries that act properly discontinuously and free of fixed points.   
\end{theorem}

The Euler characteristic of a surface $\M$ is obtained by means of their triangulations. If $V$, $E$, and $F$ are, respectively, the number of vertices, edges, and faces of a triangulation, then the Euler characteristic $\chi(\M)$ is defined as 
$\chi(\M) = V - E + F$. It is obviously not practical to obtain triangulations and count them on a surface. To work around this situation, we use the following result: If $\M$ is an orientable, compact, and connected surface of genus $g$, then $ \chi(\M)=2(1-g)$.

The hyperbolic plane's negative curvature allows for an infinite number of tessellations. In particular, while the Euclidean plane admits only three regular tessellations, the hyperbolic plane admits an an infinity quantity of such tessellations. Obviously, the existence of a regular $\{p, q\}$ tessellation  of the hyperbolic plane requires a relationship between the parameters $ p $ and $ q $. It is already established that, for $ p, q $ integers greater than 2, there exists a $ \{p,q\} $-tessellation of the hyperbolic plane if and only if $ \frac 1 p + \frac 1 q <\frac 1 2$. 

\subsection{Regular tessellations of hyperbolic surfaces}

Suppose that $\{p, q\}$ tessellates the $g$-torus $\M$ and let $\mathbb{P}$ 
be a polygon of this tessellation. If $\mu(\M)$ and $\mu(\mathbb{P})$ are, 
respectively, the areas of $ \M$ and $ \mathbb{P} $, and $F$ is the number of 
faces of this tessellation on $  \M  $, then:
\begin{equation}\label{eq:condicao-area}
	\mu (\M ) =F\mu(\mathbb{P}).
\end{equation}

Let $V$ be the number of vertices and $E$ the number of edges in this tessellation of $\M$. For a moment, consider tessellation half-edges, which are the geodesic segments obtained by half of an edge, determined by its midpoint and one of its vertices. If the number of half-edges is $E'$, then we have $E'=2E$. Since $q$ half-edges emanate from each vertex, we have $qV=E'$, that is, $qV=2E$. Similarly, if the polygons were separated from each other, each vertex would give rise to $q$  new vertices. This additional number of vertices would equal $pF$, where $ qV = pF$. Combining these relationships, we have:
\begin{equation}\label{eq:relacao-faces-vertices-arestas-tesselacao}
	qV=2E=pF.
\end{equation} 

\begin{remark}
	
	\begin{enumerate}
		
		\item The incenter of a regular $ p $-gon of a regular $ \{p,q\} $-tessellation  determines its decomposition into $ p $ congruent triangles, each with interior angles $ \frac\pi q $, $ \frac\pi q $, and $ \frac{2\pi} p $. According to the Gauss-Bonnet theorem for hyperbolic triangles, the area of each of these triangles equals $ \frac{pq-2p-2q}{pq}\pi $. Therefore, the area of the regular $ p $-gon, whose internal angles measure $ \frac{2\pi}q $, is given by $\mu(\mathbb{P})=\frac{pq-2p-2q}{ q}\pi $;
		
		\item If $ g\geq 2 $ and $ \{p,q\} $ is a regular tessellation over a $ g $-torus $  \M  $, it consists of $ F=\frac{4(g-1)q}{pq-2p-2q} $ faces, $ E=\frac{2(g-1)pq}{pq-2p-2q} $ edges, and $ V=\frac{4(g-1)p}{pq -2p-2q}$ vertices. Then, it follows that:		
		\[
		V-E+F=\frac{4(g-1)p}{pq-2p-2q} -\frac{2(g-1)pq}{pq-2p-2q} +\frac{4(g-1)q}{pq-2p-2q} =2(1-g); 
		\]
		
		\item If $ \{p,q\} $ is a regular tessellation, either of the hyperbolic plane or of a hyperbolic surface, whose edge length is $ l $, the next relation is also valid:		
		\[
		l= arcosh \displaystyle\left[ \frac{\cos^2\displaystyle\left( \frac\pi q\right) +\cos\displaystyle\left(  \frac{2\pi}{p}\right)  }{\sin ^2\displaystyle\left(  \frac\pi q \right)} \right].
		\]
	\end{enumerate}	
\end{remark}

We then get the following result.

\begin{proposition} \label{propo:equivalencia-relacoes-com-face-para-teo-kulkarni} 
	Let $\M$ be a compact, connected, and orientable surface of genus $ g\geq 2 $ and $ p,q,V,E,F $ positive integers, such that $\mu(\M) = F\mu(\mathbb{P}) $, where $\mathbb{P} $ is a $ p $-gon of interior angles measuring $ \frac{2\pi}q $. Let $F=\frac{4q(g-1)}{pq-2p-2q}$. If $ pF=2E=qV$ then $ V-E+F=\chi(\M)$. 
\end{proposition}

According to \cite{kulkarni-regular}, if $\M $ is a compact, connected, and orientable surface and if $p\geq 4 $, $q\geq 3 $, $V\geq 1 $, $E\geq 1 $, and $F\geq 1 $ are positive integers such that $V-E+F=\chi(\M) $ and $pF=2E=qV $, then $\M$ has a regular  $\{p,q\}$ tessellation.

\begin{proposition}\label{propo:medidas-l-a-r-tesselacao-regular} 	
	Given the $\{p,q\}$ tessellation of the hyperbolic plane, where $ l $ is the edge length, $ a$ and $ r$ are lengths of the radius of the circles inscribed and circumscribed to one of its faces, respectively, the following relations hold:	
	\begin{eqnarray}
		l&=& 2 arcosh\left[  \cos\displaystyle\left(  \frac{\pi}{p}\right) cosec\displaystyle\left( \frac \pi q \right) \right];   \label{formula-l-tesselacao-regular}\\
		a&=&   arccosh \left[ cosec\displaystyle\left(  \frac{\pi}{p}\right) \cos\displaystyle\left( \frac \pi q \right)  \right];  \label{formula-a-tesselacao-regular}\\		
		r&=&  arccosh\left[ cotg \displaystyle\left(  \frac{\pi}{p}\right) cotg\displaystyle\left( \frac \pi q \right) 	 \right].    \label{formula-r-tesselacao-regular}
	\end{eqnarray}
\end{proposition}

\begin{corollary}\label{cor:comprimento-geodesica-lados-opostos-poligono-que-fornece-superficie-g-toro} 
	If $ \M $ is a $ g $-torus, $ g\geq 2 $, which is obtained from a polygon of the $\{4g,4g\}$ tessellation by identifying opposite edges, the shortest cycle of non-trivial homology on this one has length $ d_\M = 2 arccosh\left[cotg \displaystyle\left( \frac{\pi}{4g}\right) \right]$.	
\end{corollary}

\begin{remark}
	The area of $\M$ is given by $\mu (\M )= 4\pi (g-1) $.
\end{remark}




\subsection{Semi-regular tessellations of surfaces derived from regular tessellations} \label{codigosNorientaveis}



%

We begin giving some definitions and results that help to guarantee the constructions of the quantum Floquet codes obtained in Section 5. For more details see \cite{artigo-giovana}.

\begin{proposition} \label{propo:ao-teo-kulkarni-r-valente-medidas} 
	Let $ \M $ be an orientable surface of genus $ g \geq 2$, on which there is a semi-regular $[m_1,m_2, m_3]$ tessellation. If $l$ denotes the edge length, $a_i$ the length of the apothem of each $m_i$-gon, $ r_i $ the length of the radius of the circle circumscribed to each $ m_i $-gon and, if $ A_i $ denotes the distance between the incenters of an $m_i$-gon and an adjacent $ m_{i+1} $-gon, $i = 1,2,3$, it follows that
	\begin{eqnarray}
		\pi &=& 	\displaystyle\sum_{i=1}^3 \arcsin\displaystyle\left(  \dfrac{\cos\displaystyle\left( \frac{\pi}{m_i}\right)}{\cosh\displaystyle\left(  \frac{l}{2}\right)}\right), 			\label{eq:teorema-tesselacoes-eq-giovana-para-lado-da-tesselacao} \\  
		a_i&=&   arcsinh\displaystyle\left(  \tanh\displaystyle\left( \frac{l}{2}\right) cotan \displaystyle\left( \frac{\pi}{m_i}\right)\right), \\
		r_i&=&    arccosh\displaystyle\left(  \cosh\displaystyle\left(  a_i\right) \cosh\displaystyle\left( \frac{l}{2}\right) \right), \\
		A_i&=&   a_i+a_{i+1}.
	\end{eqnarray}
\end{proposition}

A regular tessellation has associated the dual tessellation, formed by the union of the incenters of the adjacent faces of the original pattern. Taking into consideration that in uniform tessellations all polygons have incenters, we define an equivalent version of dual tessellation for the uniform case:

\begin{definition}
	Let $\tes$ be a uniform locally finite tessellation. We define the dual tessellation of $\tes$ as the tessellation formed by the edges that connect the incenters of the adjacent polygons present in $\tes$. If $\tes$ is not regular, the dual tessellation tile is an irregular polygon. 
\end{definition}

\begin{proposition} 
	The dual tessellation of a semi-regular tessellation $[m_1,m_2,m_3]$ is composed of triangles, not necessarily regular, whose vertices have valency $m_1,m_2,$ and $m_3$, in this order up to one cyclic permutation. The edge determined by the vertices with valence $ m_i $ and $ m_{i+1} $, $i, i+1 \in \{1,2,3\}$, whose edge lengths $ A_1, A_2$ and $A_3$, respectively, are given in the Proposition \ref{propo:ao-teo-kulkarni-r-valente-medidas}.
\end{proposition}

By utilizing specific techniques, one can construct a derived tessellation from a previously fixed one. There are three derived tessellation structures in \cite{Douglas1}. These give rise to the $ [2p,2p,q] $, and $ [2p,2q, 4] $ tessellations, starting with the regular $\{p,q\}$ tessellation. These techniques can be applied both to orientable and to non-orientable hyperbolic surfaces, provided that the respective constructions are intrinsic to a given tessellation, have a local character, and are therefore independent of the surface containing the tessellations.

Here and in the next figures, the techniques are shown in Euclidean cases, but the same ideas are used in hyperbolic cases. 

Starting from the $\{p,q\}$ tessellation, one can use the clipping derivation technique, which consists of ``cutting" a region close to each vertex of a $p$-gon so that it gives way to a regular $2p$-gon, obtaining therefore the semi-regular $[2p,2p,q]$ tessellation.

In Figure \ref*{tesselderivada}$(a)$, we show the clipping derivation technique. The dashed lines show the dual tessellation of the semi-regular tessellation. The dual tessellation of $[2p,2p,q]$ tessellation is made up of isosceles triangles where the base is twice the radius of the inscribed circle in the polygon $``2p"$ of such tessellation, and the two equal sides have the same length as the radius of the circumscribed circle in the polygon $\{p,q\}$ of the original tessellation.

Note that in this case, $n_f = f + v$, $n_e = e + q v = \frac{3}{2} p f$, $n_v = 2 e = q v = p f$, and we obtain
\[
n_f = \dfrac{(2p+2q)(g-2)}{pq-2p-2q}, \quad n_e =\dfrac{3pq(g-2)}{pq-2p-2q}
\]
and  
\[
n_v =  \dfrac{2pq(g-2)}{pq-2p-2q}.
\]

Similarly, we can derive from the $\{p, q\}$ tessellation the semi-regular $[2p, 2q, 4]$ tessellation using the Incenter Derivation process, which consists initially in drawing the radii of the inscribed and circumscribed circles of each face, respectively. This process decomposes each face into $2p$ right triangles. By drawing geodesic segments between the incenter of all pairs of triangles with common edges and discarding the edges of the original tessellation, as well as those of these triangles, one obtains the semi-regular tessellation mentioned. The dual tessellation of the $[2p,2q,4]$-tessellation is formed by right triangles whose edges are measured by the radius of the circumscribed circle, the radius of the inscribed circle, and half the edge of the polygon $\{p,q\}$ of the original tessellation. In Figure \ref*{tesselderivada}$(b)$, we show the incenter derivation technique.

\begin{figure}[H]
	\centering
		\begin{tikzpicture}[scale=.2] 
		\tiny
		\begin{scope}[scale=0.5,xshift=30cm,yshift=-7cm ]
			\foreach \cont in {1,...,6}\draw [ultra thick,rotate=\cont*60 ] (3.4641,2)-- (3.4641,-2);
			\begin{scope}[xshift=3.4641cm, yshift=6cm]\foreach \cont in {1,...,6}\draw [ultra thick,rotate=\cont*60 ] (3.4641,2)-- (3.4641,-2);\end{scope};
			\begin{scope}[xshift=6.9282cm]\foreach \cont in {1,...,6}\draw [ultra thick,rotate=\cont*60 ] (3.4641,2)-- (3.4641,-2);\end{scope};
		\end{scope}
		
		\begin{scope}[shift={(-4.8748,2.175)}]
			\foreach \ang in {0,120,240}{\begin{scope}[xshift=-12cm,rotate=\ang]
					\draw (0,0) -- (3,0);
					\draw[xshift=3cm,rotate=60] (0,0) -- (2,0);
					\draw[xshift=3cm,rotate=-60] (0,0) -- (2,0);
					\draw[yshift=1.7321cm, xshift=4cm] (0,0) -- (1,0);
					\draw[yshift=-1.7321cm, xshift=4cm] (0,0) -- (1,0);
					\draw[yshift=1.7321cm, xshift=4cm, rotate=90] (0,0) -- (1,0);
					\draw[yshift=-1.7321cm, xshift=4cm, rotate=-90] (0,0) -- (1,0);\end{scope}}		
			\foreach \ang in {0,120,240}{\begin{scope}[rotate=\ang] 
			\end{scope}};

		\end{scope}
		\begin{scope}[yshift=-9cm, xshift=-6cm, shift={(-0.375,1.425)}]
			\draw[fill=green!20!]  (1,0)--(-0.5,0.866)--(-0.5,-0.866)--cycle;
			\draw[fill=green!20!] (3.6667,1.1547)--(4.5,1.7321)--(4,2.2321)--cycle;
			\draw[fill=green!20!] (3.6667,-1.1547)--(4.5,-1.7321)--(4,-2.2321)--cycle;
			\foreach \ang in {0,120,240}
			{\begin{scope}[rotate=\ang]
					\draw[fill=purple!20!]  (4,-2.715)--(4,-2.2321) -- (3.6667,-1.155) --(3.3333,-0.57) -- (1.99,0)-- (0.995,0.01) -- (-0.495,-0.875)--  (-1.02,-1.735)-- (-1.18,-3.16)-- (-0.84,-3.75) -- (-0.095,-4.56)-- (0.38,-4.84) ;
					(3.965,-2.715)--(3.965,-2.23) -- (3.63,-1.155) --(3.31,-0.57) -- (1.99,0)-- (0.995,0.01) -- (-0.495,-0.875);
					\draw[fill=purple!20!] (5.5,1.7321) --(4.5,1.7321) -- (3.6667,1.145)-- (3.32,0.555) -- (3.305,-0.57)-- (3.66,-1.18) --(4.465,-1.73)-- (5,-1.7305) ;
					\draw[fill=green!20!] (2,0)-- (3.3333,0.5774)-- (3.3333,-0.5774)-- cycle;
					\draw[fill=green!20!] (3.6667,1.1547)--(4.5,1.7321)--(4,2.2321)--cycle;
					\draw[fill=green!20!] (3.6667,-1.1547)--(4.5,-1.7321)--(4,-2.2321)--cycle;
					\draw[dotted] (0,0) -- (3,0);
					\draw[dotted,xshift=3cm,rotate=60] (0,0) -- (2,0);
					\draw[dotted,xshift=3cm,rotate=-60] (0,0) -- (2,0);
					\draw[dotted,yshift=1.7321cm, xshift=4cm] (0,0) -- (1,0);
					\draw[dotted,yshift=-1.7321cm, xshift=4cm] (0,0) -- (1,0);
					\draw[dotted,yshift=1.7321cm, xshift=4cm, rotate=90] (0,0) -- (1,0);
					\draw[dotted,yshift=-1.7321cm, xshift=4cm, rotate=-90] (0,0) -- (1,0);
			\end{scope}};
			\node at (1.187,-5.85) {$2p$-gon};	
			\node at (4.011,-4.935) {$q$-gon};	
			\draw [-latex] (4,-1.75)   to[out=210,in=90] node [pos=.6, above] {} (4,-4.5) ;
			\draw [-latex] (1.566,-2.853)   to[out=240,in=90] node [pos=.6, above] {} (1,-5.5);
		\end{scope}
		
		\begin{scope}[xshift=8cm, yshift=-4.5cm, rotate=-120, shift={(-1.5494,-1.038)}]		

			\begin{scope}[xshift=6.9282cm]
				\draw [fill=green!20!] (2.736,-0.723)-- (2.727,0.725) -- (1.995,1.995) -- (0.73,2.728) -- (-0.714,2.728) --(-1.985,2.001) -- (-2.723,0.734)-- (-2.75,-0.718)--(-2.003,-2.007) -- (-0.733,-2.722) --(0.743,-2.72)--(2.003,-1.987) --cycle;
				\draw [fill=red!20!,rotate=120] (2.7321,0.7321) node (v3) {} rectangle (4.1962,-0.7321);
			\end{scope};
			\draw [fill=green!20!] (2.736,-0.723)-- (2.727,0.725) -- (1.995,1.995) -- (0.73,2.728) -- (-0.714,2.728) --(-1.985,2.001) -- (-2.723,0.734)-- (-2.75,-0.718)--(-2.003,-2.007) -- (-0.733,-2.722) --(0.743,-2.72)--(2.003,-1.987) --cycle;
			\begin{scope}[xshift=3.4641cm, yshift=6cm]
				\draw [fill=green!20!] (2.736,-0.723)-- (2.727,0.725) -- (1.995,1.995) -- (0.73,2.728) -- (-0.714,2.728) --(-1.985,2.001) -- (-2.723,0.734)-- (-2.75,-0.718)--(-2.003,-2.007) -- (-0.733,-2.722) --(0.743,-2.72)--(2.003,-1.987) --cycle;
				\draw [fill=red!20!,rotate=-120] (2.7321,0.7321) node (v1) {} rectangle (4.1962,-0.7321);
			\end{scope};	
			\draw [fill=red!20!] (2.7321,0.7321) rectangle (4.1962,-0.7321);
			\draw [fill=blue!20!]  (2.736,3.26) -- (2,2) -- (2.73,0.728)-- (4.19,0.728) -- (4.926,2)--(4.19,3.26) --cycle;
			\begin{scope}
				\draw [fill=red!20!] (2.7321,0.7321) rectangle (4.1962,-0.7321);
				\foreach \cont in {1,...,6}\draw [ultra thick,rotate=\cont*60 ] (3.4641,2)-- (3.4641,-2);
				\foreach \cont in {1,...,6}\draw [rotate=\cont*60,dashed] (0,0)--(3.4641,0);
				\foreach \cont in {1,...,6}\draw [rotate=\cont*60,dashed] (0,0)--(3.4641,2);
				\foreach \cont in {0,...,11}\filldraw[black, rotate=\cont*30] (2.7321,0.7321) circle (2pt) node [anchor=west] {};
				\foreach \cont in {0,...,11}\draw[rotate=\cont*30]   (2.7321,0.7321)--  (2.7321,-0.7321) ;
				\foreach \cont in {0,...,1}\draw[rotate=\cont*60]  (2.7321,0.7321) rectangle (4.1962,-0.7321);
			\end{scope};
			\begin{scope}[xshift=3.4641cm, yshift=6cm]
				\foreach \cont in {1,...,6}\draw [ultra thick,rotate=\cont*60 ] (3.4641,2)-- (3.4641,-2);
				\foreach \cont in {1,...,6}\draw [rotate=\cont*60,dashed] (0,0)--(3.4641,0);
				\foreach \cont in {1,...,6}\draw [rotate=\cont*60,dashed] (0,0)--(3.4641,2);
				\foreach \cont in {0,...,11}\filldraw[black, rotate=\cont*30] (2.7321,0.7321) circle (2pt) node [anchor=west] {};
				\foreach \cont in {0,...,11}\draw[rotate=\cont*30]   (2.7321,0.7321)--  (2.7321,-0.7321) ;
			\end{scope};
			\begin{scope}[xshift=6.9282cm]
				\foreach \cont in {1,...,6}\draw [ultra thick,rotate=\cont*60 ] (3.4641,2)-- (3.4641,-2);
				\foreach \cont in {1,...,6}\draw [rotate=\cont*60,dashed] (0,0)--(3.4641,0);
				\foreach \cont in {1,...,6}\draw [rotate=\cont*60,dashed] (0,0)--(3.4641,2);
				\foreach \cont in {0,...,11}\filldraw[black, rotate=\cont*30] (2.7321,0.7321) circle (2pt) node [anchor=west] {};
				\foreach \cont in {0,...,11}\draw[rotate=\cont*30]   (2.7321,0.7321)--  (2.7321,-0.7321) ;
			\end{scope};
			\draw (2.7321,0.7321) node (v2) {} rectangle (4.1962,-0.7321);	
			\node at (12.558,-0.4665) {$2p$-gon};
			\node at (11,3) {4-gon};
			\node at (8.3,7.3) {$2q$-gon};
			\draw [-latex] (8,0.367)  to[out=-30,in=180] node [pos=.6, above] {}(12.058,-0.4665);
			\draw [-latex] (5.8415,2.9085) to[out=-80,in=180] node [pos=.6, above] {}(10.5,3);
			\draw [-latex] (3.9085,2.6585)  to[out=120,in=220] node [pos=.6, above] {}(8,7);
			\draw [-latex, very thick] (-2,10.5)  to[out=-0,in=150] node [pos=.6, above] {}(2,10.5);
		\end{scope}
		\draw [-latex, very thick] (-14.5,-2)  to[out=-70,in=150] node[pos=.6, above] {} (-10.25,-2.75);
		
		\node at (-13,-1.75) {$(a)$};
		\node at (18.25,-8) {$(b)$};
	\end{tikzpicture}
		\caption{(a) Process of clipping derivation of the $\{p,q\}$ tessellation. In the figure on the left, the original tessellation, in the figure on the right, in solid lines, the semi-regular $[2p,2p,q]$ tessellation, and in dashed lines the dual tessellation of $[2p,2p,q]$. (b) Incenter derivation process of the $\{p,q\}$ tessellation. In the figure on the left, there is the original tessellation; in the figure on the right, in solid lines, there is the semi-regular tessellation $[2p,2q,4]$; and in dashed lines, there is the dual tessellation of $[2p,2q,4]$.} \label{tesselderivada} 
\end{figure}
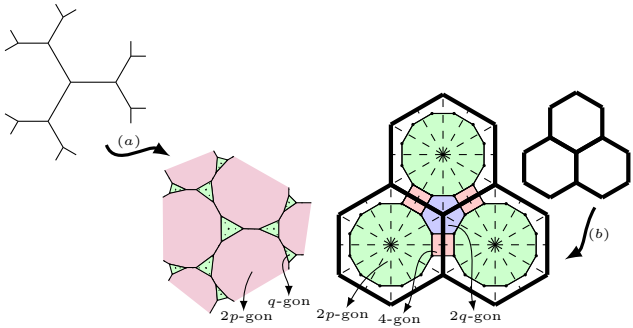

Note that in this case, $n_f = f + v + e$, $n_e = 3p f$, $n_v = 2 p f$, and we get 
\[
n_f = \dfrac{(pq+2p+2q)(g-1)}{pq-2p-2q}, \quad  n_e = \dfrac{6pq(g-1)}{pq-2p-2q} 
\]
and  
\[
n_v = \dfrac{4pq(g-1)}{pq-2p-2q}.
\]

Summarizing: if we consider a connected and non-orientable surface $M$ with genus $g\geq 3$, it is possible to construct a $\{p, q\}$ tessellation. The quantities of faces, edges, and vertices of the derived tessellations, as previously described, are given in Table \ref{tabela:N_f-N_e-N_v-tessel-derivada-em-funcao-de-p-q-g-nao-orientavel}, where $n_f = |F|$, $n_e = |E|$, and $n_v = |V|$.

\begin{table*}
	\caption{	\label{tabela:N_f-N_e-N_v-tessel-derivada-em-funcao-de-p-q-g-nao-orientavel}
		This table exhibits, as a function of $ p $, $ q $, and $ g $, the quantities $|F| $, $ |E|$, and $|V|$ relative to the semi-regular tessellations of a connected and non-orientable surface $M$, of genus $ g\geq 3 $, derived from the (regular) $\{p,q\}$ tessellation.}
	\begin{ruledtabular}
	\begin{tabular}{||c|c|c|c||}
		Tessellation	&	$ n_f $	&	$ n_e $	& $n_v$ \\	
		\hline	$[2p,2p,q]$	&	$ \dfrac{(2p+2q)(g-2)}{pq-2p-2q} $	&	$ \dfrac{3pq(g-2)}{pq-2p-2q}$	& 	$ \dfrac{2pq(g-2)}{pq-2p-2q}$ \\	
		\hline	$[2p,2q,4]$	&	$ \dfrac{(pq+2p+2q)(g-1)}{pq-2p-2q} $	&	$ \dfrac{6pq(g-1)}{pq-2p-2q} $	&	$ \dfrac{4pq(g-1)}{pq-2p-2q}$ \\	
					\end{tabular}
	\end{ruledtabular}
\end{table*}
		
\section{Quantum Floquet Codes}\label{reviewFloquet}

In this section we recall basic concepts on Floquet codes. For more details we refer the reader to \cite{Hastings2021, Vuillot2021, Breuckmann2016,Breuckmann2024}.

Hastings and Haah \cite{Hastings2021} proposed a new topological approach by means of Floquet codes that operate exclusively through measurements of two-qubit Pauli operators. Unlike Euclidean surface codes, which are constructed from four-qubit Pauli measurements, this approach eliminates the necessity to compile such measurements into a sequence of two-qubit Clifford gates and one-qubit measurements or into one and two-qubit Pauli measurements.

The logical operators evolve periodically during a syndrome extraction circuit, managing to protect the information and generate the logical qubits dynamically through a sequence of measurements. This additional flexibility allows the construction of well-performing codes~\cite{Gidney2021}, especially on platforms that excel at direct measurements of $2$ Majorana-based qubits \cite{Paetznick}.

In \cite{Paetznick}, the Euclidean semi-regular $\{4, 8, 8\}$ tessellation with boundary conditions was also used in the analysis of Floquet codes. Thus, the Kitaev model was expanded and proposed in more general geometries. Vuillot in \cite{Vuillot2021} presented a new view that allows us to explore and study Floquet codes in any geometric layout (geometric arrangement that is also $3$-colorable).

Let us now review how to construct the Instantaneous Stabilizer Group (ISG) of a Floquet code. We consider a regular trivalent $3$-colorable (hence even-sized polygons) tessellations on the hyperbolic plane whose polygons having even edges. The physical qubits are assigned to the vertices of the tessellation. The edge color depends on the face color that it lies. For instance, if a polygon is red, the edges belong to it are colored in green or blue, see Figure \ref{ISG}. 

\begin{figure}[H]
	\centering
	\label{figdouglas1}
	\begin{tikzpicture}[node distance={10mm}, thick, main/.style = {draw, circle, fill=gray, thin,inner sep=1pt}] 
		
		\def\fillgreen{[fill= green!50!]}
		\def\fillblue{[fill= blue!50!]}
		\def\fillred{[fill= red!70!]}
		\draw\fillgreen(3.7,0.8) -- (5.4,0.8) -- (5.3,1.8) -- (4.8,2.6) -- (4.1,2.7) -- (3.5,2.5) -- (3.3,2) -- (3.3,1.3) -- (3.7,0.8); 
		\draw\fillblue (5.3,1.8) -- (5.4,0.8) -- (6.2,0.1) -- (7.1,0.6) -- (7.5,1.5) -- (7.2,2.4) -- (6.6,3) -- (5.7,2.8) -- (5.3,1.8); 
		\draw\fillred (7.1,0.6) -- (8,0.4) -- (9.1,0.5) -- (9.6,1.05) -- (9.45,1.8) -- (8.9,2.3) -- (8.1,2.1) -- (7.5,1.5) -- (7.1,0.6); 
		\draw \fillgreen (6.2,-1.2) -- (7,-1.4) -- (7.95,-1.1) -- (8.5,-0.65) -- (8.5,0) -- (8,0.4) -- (7.1,0.6) -- (6.2,0.1) -- (6.2,-1.2); 
		\draw\fillred(6.2,0.1) -- (6.2,-1.2) -- (5.2,-2) -- (3.7,-1.7) -- (3.3,-1) -- (3.2,0.2) -- (3.7,0.8) -- (5.4,0.8) -- (6.2,0.1); 
		\draw\fillgreen (3.7,-1.7) -- (5.2,-2) -- (5.4,-2.9) -- (4.9,-3.6) -- (4,-3.7) -- (3.2,-3.6) -- (2.85,-3.2) -- (3.15,-2.2) -- (3.7,-1.7); 
		\draw\fillblue (5.2,-2) -- (5.4,-2.9) -- (5.9,-3.35) -- (6.5,-3.4) -- (6.95,-3.1) -- (7.2,-2.3) -- (7,-1.4) -- (6.2,-1.2) -- (5.2,-2); 
		\draw\fillred (7.2,-2.3) -- (8.2,-2.8) -- (8.8,-2.55) -- (9,-2) -- (8.95,-1.45) -- (8.55,-1.15) -- (7.95,-1.1) -- (7,-1.4) -- (7.2,-2.3); 
		
		\node[main] (1) at (3.7,-1.7) {};
		\node[main] (2) at (5.2,-2) {};
		\node[main] (3) at (6.2,-1.2) {};
		\node[main] (4) at (6.2,0.1) {};
		\node[main] (5) at (5.4,0.8) {};,
		\node[main] (6) at (3.7,0.8) {};
		\node[main] (7) at (3.2,0.2) {};
		\node[main] (8) at (3.3,-1) {};
		\node[main] (12) at (7.95,-1.1) {};
		\node[main] (13) at (8.5,-0.65) {};
		\node[main] (15) at (8.5,0) {};
		\node[main] (16) at (8,0.4) {};
		\node[main] (9) at (7.1,0.6) {};
		\node[main] (10) at (7,-1.4) {};
		\node[main] (17) at (5.4,-2.9) {};
		\node[main] (18) at (5.9,-3.35) {};
		\node[main] (19) at (6.5,-3.4) {};
		\node[main] (20) at (7.2,-2.3) {};
		\node[main] (21) at (5.3,1.8)   {};
		\node[main] (22) at (5.7,2.8) {};
		\node[main] (23) at (6.6,3) {};
		\node[main] (24) at (7.2,2.4) {};
		\node[main] (25) at (7.5,1.5) {};
		\node[main] (26) at (9.1,0.5) {};
		\node[main] (27) at (9.6,1.05) {};
		\node[main] (28) at (9.45,1.8) {};
		\node[main] (29) at (8.9,2.3) {};
		\node[main] (30) at (8.1,2.1) {};
		\node[main] (31) at (4.8,2.6) {};
		\node[main] (32) at (4.1,2.7) {};
		\node[main] (33) at (3.5,2.5) {};
		\node[main] (34) at (3.3,2) {};
		\node[main] (35) at (3.3,1.3) {};
		\node[main] (36) at (3.15,-2.2) {};
		\node[main] (37) at (2.85,-3.2) {};
		\node[main] (38) at (3.2,-3.6) {};
		\node[main] (39) at (4,-3.7) {};
		\node[main] (40) at (4.9,-3.6) {};
		\node[main] (41) at (8.2,-2.8) {};
		\node[main] (42) at (8.8,-2.55) {};
		\node[main] (43) at (9,-2) {};
		\node[main] (44) at (8.95,-1.45) {};
		\node[main] (45) at (8.55,-1.15) {};
		\node[main] (46) at (6.95,-3.1) {};
		
		\begin{scope}[very thick]
			\def\arestagreen{[very thick, green]}
			\def\arestablue{[very thick, blue]}
			\def\arestared{[very thick, red]}
			\draw\arestagreen (8)--(1);
			\draw\arestablue (1)--(2);
			\draw\arestagreen (2)--(3);
			\draw\arestablue (3)--(4);
			\draw \arestagreen(4)--(5);
			\draw\arestablue (5)--(6);
			\draw\arestagreen (6)--(7);
			\draw\arestablue (7)--(8);	
			\draw\arestared(4)--(9);
			\draw \arestared(3)--(10);
			\draw \arestagreen(10)--(20);
			\draw\arestablue (10)--(12);
			\draw\arestared (13)--(12);
			\draw \arestablue(13)--(15);
			\draw\arestared (15)--(16);
			\draw\arestablue (9)--(16);
			\draw\arestared (2)--(17);
			\draw\arestagreen (17)--(18);
			\draw\arestared (18)--(19);
			\draw\arestagreen (19)--(46);
			\draw\arestared  (46)--(20);
			\draw \arestared (5)--(21);
			\draw\arestagreen (21)--(22);
			\draw\arestared  (22)--(23);
			\draw\arestagreen (23)--(24);
			\draw \arestared (24)--(25);
			\draw\arestagreen (25)--(9);
			\draw\arestared  (1)--(36);
			\draw\arestablue (36)--(37);
			\draw\arestared  (37)--(38);
			\draw\arestablue (38)--(39);
			\draw\arestared  (39)--(40);
			\draw\arestablue (40)--(17);
			\draw\arestablue (21)--(31);
			\draw\arestared  (31)--(32);
			\draw\arestablue (32)--(33);
			\draw \arestared (33)--(34);
			\draw \arestablue(34)--(35);
			\draw\arestared  (35)--(6);
			\draw\arestagreen (16)--(26);
			\draw \arestablue(26)--(27);
			\draw \arestagreen(27)--(28);
			\draw\arestablue (28)--(29);
			\draw \arestagreen(29)--(30);
			\draw \arestablue(30)--(25);
			\draw \arestablue(20)--(41);
			\draw \arestagreen(41)--(42);
			\draw \arestablue(42)--(43);
			\draw\arestagreen (43)--(44);
			\draw \arestablue(44)--(45);
			\draw \arestagreen(45)--(12);
		\end{scope}
	\end{tikzpicture}
	\caption{An example of a regular three colorable tessellation.}
	\label{ISG}
\end{figure}
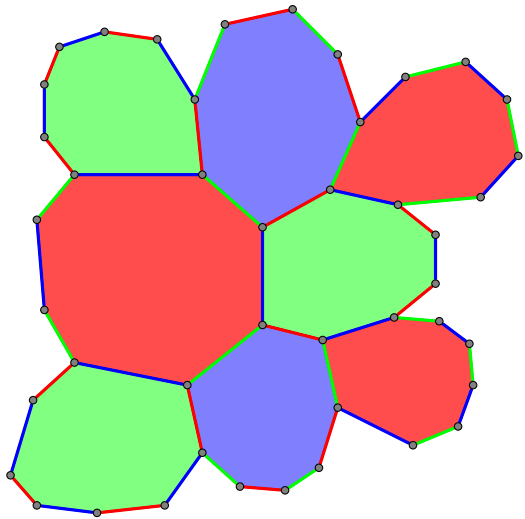

The checks are defined by $XX$, $YY$ and $ZZ$. Every check has a color: the green, blue and red edges correspond to $XX$, $YY$, $ZZ$ measurement checks, respectively, up to the signal measurement. Note also that green (blue, red) edges connect red (blue, green) faces. Two faces of different colors intersect only in edges of the another color. For instance, a green face intersects a blue face 
only in red edges. The three-coloured edges (the checks) will be utilized to perform 
the repetitive rounds of these two-body check measurements of the qubits of the edges. More precisely, at rounds $r=3n$ $($r=3n+1$, $r=3n+2$)$, all green (blue, red) checks are measured, see Figure \ref{figdouglas2}. 

\begin{figure}[H]
	\centering
	\label{figdouglas2}
\begin{tikzpicture}[scale=0.52]
	\tiny
	\def\qubitcolor{black}
	\def\qubitfill{gray!50!}
	\def\configred{red!50!}
	\def\configblue{blue!30!}
	\def\configgreen{green!50!}
	\def\espessuralinha{very thick}
	\def\qubitraio{0.06}
	\def\qubitsuporteraio{0.03}
	\def\qubitsuportecolor{black}
	\def\qubitsuportefill{gray!30!}
	
	\begin{scope}[shift={(-5.4,0)}]              
		\clip  (2.4495,2.1213) node (v1) {} rectangle (7.3485,-2.1213) node (v2) {};
		\def\principal{ \begin{scope}
				\draw[fill=\configgreen] (0.2588,0.9659)--(-0.2588,0.9659)--(-0.7071,0.7071)--(-0.9659,0.2588)--(-0.9659,-0.2588)--(-0.7071,-0.7071)--(-0.2588,-0.9659)--(0.2588,-0.9659)--(-0.2588,-0.9659)--(0.2588,-0.9659)--(0.7071,-0.7071)--(0.9659,-0.2588)--(0.9659,0.2588)--(0.7071,0.7071)--cycle;
				\draw[fill=\configred]  (0.9659,-0.2588)-- (1.4835,-0.2588)-- (1.4835,0.2588)-- (0.9659,0.2588)-- (0.9659,-0.2588);
				\draw[fill=\configred]  (0.7071,0.7071)-- (0.9659,1.1553)-- (0.5176,1.4142)-- (0.2588,0.9659)-- (0.7071,0.7071);
				\draw[fill=\configred]  (-0.2588,0.9659)-- (-0.5176,1.4142)-- (-0.9659,1.1553)-- (-0.7071,0.7071)-- (-0.2588,0.9659);
				\draw[fill=\configred]  (-0.9659,0.2588)-- (-1.4835,0.2588)-- (-1.4835,-0.2588)-- (-0.9659,-0.2588)-- (-0.9659,0.2588);
				\draw[fill=\configred]  (-0.7071,-0.7071)-- (-0.9659,-1.1553)-- (-0.5176,-1.4142)-- (-0.2588,-0.9659)-- (-0.7071,-0.7071);
				\draw[fill=\configred]  (0.2588,-0.9659)-- (0.5176,-1.4142)-- (0.9659,-1.1553)-- (0.7071,-0.7071)-- (0.2588,-0.9659);
				\draw[fill=\configred]  (0.9659,-0.2588)-- (1.4835,-0.2588)-- (1.4835,0.2588)-- (0.9659,0.2588)-- (0.9659,-0.2588);
				\draw[blue!30!, fill=\configblue]  (0.2588,0.9659)-- (0.5176,1.4142)-- (-0.5176,1.4142)-- (-0.2588,0.9659)--cycle;
				\draw[blue!30!, fill=\configblue]  (-0.7071,0.7071)-- (-0.9659,1.1553)-- (-1.4835,0.2588)-- (-0.9659,0.2588)--cycle;
				\draw[blue!30!, fill=\configblue]  (-0.9659,-0.2588)-- (-1.4835,-0.2588)-- (-0.9659,-1.1553)-- (-0.7071,-0.7071)--cycle;
				\draw[blue!30!, fill=\configblue]  (-0.2588,-0.9659)-- (-0.5176,-1.4142)-- (0.5176,-1.4142)-- (0.2588,-0.9659)--cycle;
				\draw[blue!30!, fill=\configblue]  (0.7071,-0.7071)-- (0.9659,-1.1553)-- (1.4835,-0.2588)-- (0.9659,-0.2588)--cycle;
				\draw[blue!30!, fill=\configblue]  (0.9659,0.2588)-- (1.4835,0.2588)-- (0.9659,1.1553)-- (0.7071,0.7071)--cycle;
				\draw[red,\espessuralinha]  (-0.7071,0.7071)-- (-0.9659,0.2588);
				\draw[red, \espessuralinha]  (-0.9659,-0.2588)-- (-0.7071,-0.7071);
				\draw[red,\espessuralinha]  (-0.2588,-0.9659)-- (0.2588,-0.9659);
				\draw[red,\espessuralinha]  (0.7071,-0.7071)-- (0.9659,-0.2588);
				\draw[red,\espessuralinha]  (0.9659,0.2588)-- (0.7071,0.7071);
				\draw[red,\espessuralinha]  (0.2588,0.9659)-- (-0.2588,0.9659);
				\draw[blue,\espessuralinha]  (-0.9659,0.2588)-- (-0.9659,-0.2588);
				\draw[blue,\espessuralinha]  (-0.7071,-0.7071)-- (-0.2588,-0.9659);
				\draw[blue,\espessuralinha]  (0.2588,-0.9659)-- (0.7071,-0.7071);
				\draw[blue,\espessuralinha]  (0.9659,-0.2588)-- (0.9659,0.2588);
				\draw[blue,\espessuralinha]  (0.7071,0.7071)-- (0.2588,0.9659);
				\draw[blue,\espessuralinha]  (-0.2588,0.9659)-- (-0.7071,0.7071);
				\draw[fill=\configred]  (0.9659,-0.2588)-- (1.4835,-0.2588)-- (1.4835,0.2588)-- (0.9659,0.2588)-- (0.9659,-0.2588);
				\draw[green,\espessuralinha]  (0.7071,0.7071)-- (0.9659,1.1553);
				\draw[green,\espessuralinha]  (0.2588,0.9659)-- (0.5176,1.4142);
				\draw[blue,\espessuralinha]  (0.9659,1.1553)-- (0.5176,1.4142);
				\draw[green,\espessuralinha]  (-0.2588,0.9659)-- (-0.5176,1.4142);
				\draw[green,\espessuralinha]  (-0.7071,0.7071)-- (-0.9659,1.1553);
				\draw[blue,\espessuralinha]  (-0.5176,1.4142)-- (-0.9659,1.1553);
				\draw[green,\espessuralinha]  (0.2588,-0.9659)-- (0.5176,-1.4142);
				\draw[green,\espessuralinha]  (-0.2588,-0.9659)-- (-0.5176,-1.4142);
				\draw[blue,\espessuralinha]  (-0.9659,-1.1553)-- (-0.5176,-1.4142);
				\draw[green,\espessuralinha]  (-0.7071,-0.7071)-- (-0.9659,-1.1553);
				\draw[green,\espessuralinha]  (-0.9659,-0.2588)-- (-1.4835,-0.2588);
				\draw[blue,\espessuralinha]  (-1.4835,0.2588)-- (-1.4835,-0.2588);
				\draw[green,\espessuralinha]  (-0.9659,0.2588)-- (-1.4835,0.2588);
				\draw[green,\espessuralinha]  (0.7071,-0.7071)-- (0.9659,-1.1553);
				\draw[blue,\espessuralinha]  (0.5176,-1.4142)-- (0.9659,-1.1553);
				\draw[green,\espessuralinha]  (0.9659,-0.2588)-- (1.4835,-0.2588);
				\draw[green,\espessuralinha]  (0.9659,0.2588)-- (1.4835,0.2588);
				\draw[blue,\espessuralinha]  (1.4835,-0.2588)-- (1.4835,0.2588);
				\draw[\qubitcolor, fill=\qubitfill]  (-0.2588,0.9659) circle (\qubitraio);\draw[\qubitcolor, fill=\qubitfill]  (-0.9659,0.2588) circle (\qubitraio);\draw[\qubitcolor, fill=\qubitfill] (-0.7071,0.7071)circle (\qubitraio);
				\draw[\qubitcolor, fill=\qubitfill]  (-0.9659,-0.2588) circle (\qubitraio);\draw[\qubitcolor, fill=\qubitfill]  (-0.7071,-0.7071) circle (\qubitraio);\draw[\qubitcolor, fill=\qubitfill]  (-0.2588,-0.9659) circle (\qubitraio);\draw[\qubitcolor, fill=\qubitfill]  (0.2588,-0.9659) circle (\qubitraio);\draw[\qubitcolor, fill=\qubitfill]  (0.7071,-0.7071) circle (\qubitraio);\draw[\qubitcolor, fill=\qubitfill]  (0.9659,-0.2588) circle (\qubitraio);\draw[\qubitcolor, fill=\qubitfill]  (0.9659,0.2588) circle (\qubitraio);\draw[\qubitcolor, fill=\qubitfill]  (0.7071,0.7071) circle (\qubitraio);\draw[\qubitcolor, fill=\qubitfill]  (-0.5176,1.4142) circle (\qubitraio);\draw[\qubitcolor, fill=\qubitfill]  (-0.9659,1.1553) circle (\qubitraio);\draw[\qubitcolor, fill=\qubitfill]  (-1.4835,0.2588) circle (\qubitraio);\draw[\qubitcolor, fill=\qubitfill]  (-1.4835,-0.2588) circle (\qubitraio);\draw[\qubitcolor, fill=\qubitfill]  (-0.9659,-1.1553) circle (\qubitraio);\draw[\qubitcolor, fill=\qubitfill]  (-0.5176,-1.4142) circle (\qubitraio);\draw[\qubitcolor, fill=\qubitfill]  (0.5176,-1.4142) circle (\qubitraio);\draw[\qubitcolor, fill=\qubitfill]  (0.9659,-1.1553) circle (\qubitraio);\draw[\qubitcolor, fill=\qubitfill]  (1.4835,-0.2588) circle (\qubitraio);\draw[\qubitcolor, fill=\qubitfill]  (1.4835,0.2588) circle (\qubitraio);\draw[\qubitcolor, fill=\qubitfill]  (0.9659,1.1553) circle (\qubitraio);\draw[\qubitcolor, fill=\qubitfill]  (0.5176,1.4142) circle (\qubitraio);
		\end{scope}}
		\foreach \x in {2.44948974278318 , 4.89897948556636 , 7.34846922834953 , 9.79795897113271 , 12.2474487139159 } {\begin{scope}[xshift=\x cm]  \principal \end{scope} }
		\foreach \x in {1.22474487139159 , 3.67423461417477 , 6.12372435695795 , 8.57321409974112 , 11.0227038425243 , 13.4721935853075 } {\begin{scope}[xshift=\x cm, yshift=2.1213cm]  \principal \end{scope} }
		\foreach \x in {1.22474487139159 , 3.67423461417477 , 6.12372435695795 , 8.57321409974112 , 11.0227038425243 , 13.4721935853075 } {\begin{scope}[xshift=\x cm, yshift=-2.1213cm]  \principal \end{scope} }
		
		\draw[\qubitcolor, fill=\qubitfill]  (6.3815,-1.1565) circle (\qubitraio); 
		\draw[\qubitcolor, fill=\qubitfill] (3.934,-1.158)  circle (\qubitraio); 
		\draw[draw opacity=0.3,fill opacity=0.5,text opacity=0.2, line width=1.5ex]  plot[smooth, tension=0.1] coordinates {(4.64,-2.19)  (4.64,-1.86) (4.38,-1.413) (3.934,-1.158) (4.191,-0.705) (3.933,-0.26) (3.933,0.26) (4.194,0.711) (4.64,0.965) (4.382,1.41) (4.64,1.859) (4.64,2.18)};
		\draw[\qubitsuportecolor, fill=\qubitsuportefill]  (4.64,1.859) circle (\qubitsuporteraio);
		\draw[\qubitsuportecolor, fill=\qubitsuportefill]  (4.3815,1.4105) circle (\qubitsuporteraio);
		\draw[\qubitsuportecolor, fill=\qubitsuportefill]  (4.194,0.711) circle (\qubitsuporteraio);
		\draw[\qubitsuportecolor, fill=\qubitsuportefill]  (3.933,0.26) circle (\qubitsuporteraio);
		\draw[\qubitsuportecolor, fill=\qubitsuportefill]  (3.933,-0.26) circle (\qubitsuporteraio);
		\draw[\qubitsuportecolor, fill=\qubitsuportefill]  (4.191,-0.705) circle (\qubitsuporteraio);
		\draw[\qubitsuportecolor, fill=\qubitsuportefill]  (4.38,-1.413) circle (\qubitsuporteraio);
		\draw[\qubitsuportecolor, fill=\qubitsuportefill]  (4.64,-1.86) circle (\qubitsuporteraio);
	\end{scope}
	
	\begin{scope}[shift={(-0.446,0)}]              
		\clip  (2.4495,2.1213) rectangle (7.3485,-2.1213);
		\def\principal{ \begin{scope}
				\draw[fill=\configgreen] (0.2588,0.9659)--(-0.2588,0.9659)--(-0.7071,0.7071)--(-0.9659,0.2588)--(-0.9659,-0.2588)--(-0.7071,-0.7071)--(-0.2588,-0.9659)--(0.2588,-0.9659)--(-0.2588,-0.9659)--(0.2588,-0.9659)--(0.7071,-0.7071)--(0.9659,-0.2588)--(0.9659,0.2588)--(0.7071,0.7071)--cycle;
				\draw[fill=\configred]  (0.9659,-0.2588)-- (1.4835,-0.2588)-- (1.4835,0.2588)-- (0.9659,0.2588)-- (0.9659,-0.2588);
				\draw[fill=\configred]  (0.7071,0.7071)-- (0.9659,1.1553)-- (0.5176,1.4142)-- (0.2588,0.9659)-- (0.7071,0.7071);
				\draw[fill=\configred]  (-0.2588,0.9659)-- (-0.5176,1.4142)-- (-0.9659,1.1553)-- (-0.7071,0.7071)-- (-0.2588,0.9659);
				\draw[fill=\configred]  (-0.9659,0.2588)-- (-1.4835,0.2588)-- (-1.4835,-0.2588)-- (-0.9659,-0.2588)-- (-0.9659,0.2588);
				\draw[fill=\configred]  (-0.7071,-0.7071)-- (-0.9659,-1.1553)-- (-0.5176,-1.4142)-- (-0.2588,-0.9659)-- (-0.7071,-0.7071);
				\draw[fill=\configred]  (0.2588,-0.9659)-- (0.5176,-1.4142)-- (0.9659,-1.1553)-- (0.7071,-0.7071)-- (0.2588,-0.9659);
				\draw[fill=\configred]  (0.9659,-0.2588)-- (1.4835,-0.2588)-- (1.4835,0.2588)-- (0.9659,0.2588)-- (0.9659,-0.2588);
				\draw[blue!30!, fill=\configblue]  (0.2588,0.9659)-- (0.5176,1.4142)-- (-0.5176,1.4142)-- (-0.2588,0.9659)--cycle;
				\draw[blue!30!, fill=\configblue]  (-0.7071,0.7071)-- (-0.9659,1.1553)-- (-1.4835,0.2588)-- (-0.9659,0.2588)--cycle;
				\draw[blue!30!, fill=\configblue]  (-0.9659,-0.2588)-- (-1.4835,-0.2588)-- (-0.9659,-1.1553)-- (-0.7071,-0.7071)--cycle;
				\draw[blue!30!, fill=\configblue]  (-0.2588,-0.9659)-- (-0.5176,-1.4142)-- (0.5176,-1.4142)-- (0.2588,-0.9659)--cycle;
				\draw[blue!30!, fill=\configblue]  (0.7071,-0.7071)-- (0.9659,-1.1553)-- (1.4835,-0.2588)-- (0.9659,-0.2588)--cycle;
				\draw[blue!30!, fill=\configblue]  (0.9659,0.2588)-- (1.4835,0.2588)-- (0.9659,1.1553)-- (0.7071,0.7071)--cycle;
				\draw[red,\espessuralinha]  (-0.7071,0.7071)-- (-0.9659,0.2588);
				\draw[red, \espessuralinha]  (-0.9659,-0.2588)-- (-0.7071,-0.7071);
				\draw[red,\espessuralinha]  (-0.2588,-0.9659)-- (0.2588,-0.9659);
				\draw[red,\espessuralinha]  (0.7071,-0.7071)-- (0.9659,-0.2588);
				\draw[red,\espessuralinha]  (0.9659,0.2588)-- (0.7071,0.7071);
				\draw[red,\espessuralinha]  (0.2588,0.9659)-- (-0.2588,0.9659);
				\draw[blue,\espessuralinha]  (-0.9659,0.2588)-- (-0.9659,-0.2588);
				\draw[blue,\espessuralinha]  (-0.7071,-0.7071)-- (-0.2588,-0.9659);
				\draw[blue,\espessuralinha]  (0.2588,-0.9659)-- (0.7071,-0.7071);
				\draw[blue,\espessuralinha]  (0.9659,-0.2588)-- (0.9659,0.2588);
				\draw[blue,\espessuralinha]  (0.7071,0.7071)-- (0.2588,0.9659);
				\draw[blue,\espessuralinha]  (-0.2588,0.9659)-- (-0.7071,0.7071);
				\draw[fill=\configred]  (0.9659,-0.2588)-- (1.4835,-0.2588)-- (1.4835,0.2588)-- (0.9659,0.2588)-- (0.9659,-0.2588);
				\draw[green,\espessuralinha]  (0.7071,0.7071)-- (0.9659,1.1553);
				\draw[green,\espessuralinha]  (0.2588,0.9659)-- (0.5176,1.4142);
				\draw[blue,\espessuralinha]  (0.9659,1.1553)-- (0.5176,1.4142);
				\draw[green,\espessuralinha]  (-0.2588,0.9659)-- (-0.5176,1.4142);
				\draw[green,\espessuralinha]  (-0.7071,0.7071)-- (-0.9659,1.1553);
				\draw[blue,\espessuralinha]  (-0.5176,1.4142)-- (-0.9659,1.1553);
				\draw[green,\espessuralinha]  (0.2588,-0.9659)-- (0.5176,-1.4142);
				\draw[green,\espessuralinha]  (-0.2588,-0.9659)-- (-0.5176,-1.4142);
				\draw[blue,\espessuralinha]  (-0.9659,-1.1553)-- (-0.5176,-1.4142);
				\draw[green,\espessuralinha]  (-0.7071,-0.7071)-- (-0.9659,-1.1553);
				\draw[green,\espessuralinha]  (-0.9659,-0.2588)-- (-1.4835,-0.2588);
				\draw[blue,\espessuralinha]  (-1.4835,0.2588)-- (-1.4835,-0.2588);
				\draw[green,\espessuralinha]  (-0.9659,0.2588)-- (-1.4835,0.2588);
				\draw[green,\espessuralinha]  (0.7071,-0.7071)-- (0.9659,-1.1553);
				\draw[blue,\espessuralinha]  (0.5176,-1.4142)-- (0.9659,-1.1553);
				\draw[green,\espessuralinha]  (0.9659,-0.2588)-- (1.4835,-0.2588);
				\draw[green,\espessuralinha]  (0.9659,0.2588)-- (1.4835,0.2588);
				\draw[blue,\espessuralinha]  (1.4835,-0.2588)-- (1.4835,0.2588);
				\draw[\qubitcolor, fill=\qubitfill]  (-0.2588,0.9659) circle (\qubitraio);\draw[\qubitcolor, fill=\qubitfill]  (-0.9659,0.2588) circle (\qubitraio);\draw[\qubitcolor, fill=\qubitfill] (-0.7071,0.7071)circle (\qubitraio);
				\draw[\qubitcolor, fill=\qubitfill]  (-0.9659,-0.2588) circle (\qubitraio);\draw[\qubitcolor, fill=\qubitfill]  (-0.7071,-0.7071) circle (\qubitraio);\draw[\qubitcolor, fill=\qubitfill]  (-0.2588,-0.9659) circle (\qubitraio);\draw[\qubitcolor, fill=\qubitfill]  (0.2588,-0.9659) circle (\qubitraio);\draw[\qubitcolor, fill=\qubitfill]  (0.7071,-0.7071) circle (\qubitraio);\draw[\qubitcolor, fill=\qubitfill]  (0.9659,-0.2588) circle (\qubitraio);\draw[\qubitcolor, fill=\qubitfill]  (0.9659,0.2588) circle (\qubitraio);\draw[\qubitcolor, fill=\qubitfill]  (0.7071,0.7071) circle (\qubitraio);\draw[\qubitcolor, fill=\qubitfill]  (-0.5176,1.4142) circle (\qubitraio);\draw[\qubitcolor, fill=\qubitfill]  (-0.9659,1.1553) circle (\qubitraio);\draw[\qubitcolor, fill=\qubitfill]  (-1.4835,0.2588) circle (\qubitraio);\draw[\qubitcolor, fill=\qubitfill]  (-1.4835,-0.2588) circle (\qubitraio);\draw[\qubitcolor, fill=\qubitfill]  (-0.9659,-1.1553) circle (\qubitraio);\draw[\qubitcolor, fill=\qubitfill]  (-0.5176,-1.4142) circle (\qubitraio);\draw[\qubitcolor, fill=\qubitfill]  (0.5176,-1.4142) circle (\qubitraio);\draw[\qubitcolor, fill=\qubitfill]  (0.9659,-1.1553) circle (\qubitraio);\draw[\qubitcolor, fill=\qubitfill]  (1.4835,-0.2588) circle (\qubitraio);\draw[\qubitcolor, fill=\qubitfill]  (1.4835,0.2588) circle (\qubitraio);\draw[\qubitcolor, fill=\qubitfill]  (0.9659,1.1553) circle (\qubitraio);\draw[\qubitcolor, fill=\qubitfill]  (0.5176,1.4142) circle (\qubitraio);
		\end{scope}}
		\foreach \x in {2.44948974278318 , 4.89897948556636 , 7.34846922834953 , 9.79795897113271 , 12.2474487139159 } {\begin{scope}[xshift=\x cm]  \principal \end{scope} }
		\foreach \x in {1.22474487139159 , 3.67423461417477 , 6.12372435695795 , 8.57321409974112 , 11.0227038425243 , 13.4721935853075 } {\begin{scope}[xshift=\x cm, yshift=2.1213cm]  \principal \end{scope} }
		\foreach \x in {1.22474487139159 , 3.67423461417477 , 6.12372435695795 , 8.57321409974112 , 11.0227038425243 , 13.4721935853075 } {\begin{scope}[xshift=\x cm, yshift=-2.1213cm]  \principal \end{scope} }
		\draw[\qubitcolor, fill=\qubitfill]  (6.3815,-1.1565) circle (\qubitraio); 
		\draw[\qubitcolor, fill=\qubitfill] (3.934,-1.158)  circle (\qubitraio); 
		\draw[draw opacity=0.3,fill opacity=0.5,text opacity=0.2, line width=1.5ex]  plot[smooth, tension=0.1] coordinates {(4.64,-2.19)  (4.64,-1.86) (4.38,-1.413) (3.934,-1.158) (4.191,-0.705) (3.933,-0.26) (3.933,0.26) (4.194,0.711) (4.64,0.965) (4.382,1.41) (4.64,1.859) (4.64,2.18)};
	\end{scope}
	\begin{scope}[shift={(4.5066,0)}]              
		\clip  (2.4495,2.1213) rectangle (7.3485,-2.1213);
		\def\principal{ \begin{scope}
				\draw[fill=\configgreen] (0.2588,0.9659)--(-0.2588,0.9659)--(-0.7071,0.7071)--(-0.9659,0.2588)--(-0.9659,-0.2588)--(-0.7071,-0.7071)--(-0.2588,-0.9659)--(0.2588,-0.9659)--(-0.2588,-0.9659)--(0.2588,-0.9659)--(0.7071,-0.7071)--(0.9659,-0.2588)--(0.9659,0.2588)--(0.7071,0.7071)--cycle;
				\draw[fill=\configred]  (0.9659,-0.2588)-- (1.4835,-0.2588)-- (1.4835,0.2588)-- (0.9659,0.2588)-- (0.9659,-0.2588);
				\draw[fill=\configred]  (0.7071,0.7071)-- (0.9659,1.1553)-- (0.5176,1.4142)-- (0.2588,0.9659)-- (0.7071,0.7071);
				\draw[fill=\configred]  (-0.2588,0.9659)-- (-0.5176,1.4142)-- (-0.9659,1.1553)-- (-0.7071,0.7071)-- (-0.2588,0.9659);
				\draw[fill=\configred]  (-0.9659,0.2588)-- (-1.4835,0.2588)-- (-1.4835,-0.2588)-- (-0.9659,-0.2588)-- (-0.9659,0.2588);
				\draw[fill=\configred]  (-0.7071,-0.7071)-- (-0.9659,-1.1553)-- (-0.5176,-1.4142)-- (-0.2588,-0.9659)-- (-0.7071,-0.7071);
				\draw[fill=\configred]  (0.2588,-0.9659)-- (0.5176,-1.4142)-- (0.9659,-1.1553)-- (0.7071,-0.7071)-- (0.2588,-0.9659);
				\draw[fill=\configred]  (0.9659,-0.2588)-- (1.4835,-0.2588)-- (1.4835,0.2588)-- (0.9659,0.2588)-- (0.9659,-0.2588);
				\draw[blue!30!, fill=\configblue]  (0.2588,0.9659)-- (0.5176,1.4142)-- (-0.5176,1.4142)-- (-0.2588,0.9659)--cycle;
				\draw[blue!30!, fill=\configblue]  (-0.7071,0.7071)-- (-0.9659,1.1553)-- (-1.4835,0.2588)-- (-0.9659,0.2588)--cycle;
				\draw[blue!30!, fill=\configblue]  (-0.9659,-0.2588)-- (-1.4835,-0.2588)-- (-0.9659,-1.1553)-- (-0.7071,-0.7071)--cycle;
				\draw[blue!30!, fill=\configblue]  (-0.2588,-0.9659)-- (-0.5176,-1.4142)-- (0.5176,-1.4142)-- (0.2588,-0.9659)--cycle;
				\draw[blue!30!, fill=\configblue]  (0.7071,-0.7071)-- (0.9659,-1.1553)-- (1.4835,-0.2588)-- (0.9659,-0.2588)--cycle;
				\draw[blue!30!, fill=\configblue]  (0.9659,0.2588)-- (1.4835,0.2588)-- (0.9659,1.1553)-- (0.7071,0.7071)--cycle;
				\draw[red,\espessuralinha]  (-0.7071,0.7071)-- (-0.9659,0.2588);
				\draw[red, \espessuralinha]  (-0.9659,-0.2588)-- (-0.7071,-0.7071);
				\draw[red,\espessuralinha]  (-0.2588,-0.9659)-- (0.2588,-0.9659);
				\draw[red,\espessuralinha]  (0.7071,-0.7071)-- (0.9659,-0.2588);
				\draw[red,\espessuralinha]  (0.9659,0.2588)-- (0.7071,0.7071);
				\draw[red,\espessuralinha]  (0.2588,0.9659)-- (-0.2588,0.9659);
				\draw[blue,\espessuralinha]  (-0.9659,0.2588)-- (-0.9659,-0.2588);
				\draw[blue,\espessuralinha]  (-0.7071,-0.7071)-- (-0.2588,-0.9659);
				\draw[blue,\espessuralinha]  (0.2588,-0.9659)-- (0.7071,-0.7071);
				\draw[blue,\espessuralinha]  (0.9659,-0.2588)-- (0.9659,0.2588);
				\draw[blue,\espessuralinha]  (0.7071,0.7071)-- (0.2588,0.9659);
				\draw[blue,\espessuralinha]  (-0.2588,0.9659)-- (-0.7071,0.7071);
				\draw[fill=\configred]  (0.9659,-0.2588)-- (1.4835,-0.2588)-- (1.4835,0.2588)-- (0.9659,0.2588)-- (0.9659,-0.2588);
				\draw[green,\espessuralinha]  (0.7071,0.7071)-- (0.9659,1.1553);
				\draw[green,\espessuralinha]  (0.2588,0.9659)-- (0.5176,1.4142);
				\draw[blue,\espessuralinha]  (0.9659,1.1553)-- (0.5176,1.4142);
				\draw[green,\espessuralinha]  (-0.2588,0.9659)-- (-0.5176,1.4142);
				\draw[green,\espessuralinha]  (-0.7071,0.7071)-- (-0.9659,1.1553);
				\draw[blue,\espessuralinha]  (-0.5176,1.4142)-- (-0.9659,1.1553);
				\draw[green,\espessuralinha]  (0.2588,-0.9659)-- (0.5176,-1.4142);
				\draw[green,\espessuralinha]  (-0.2588,-0.9659)-- (-0.5176,-1.4142);
				\draw[blue,\espessuralinha]  (-0.9659,-1.1553)-- (-0.5176,-1.4142);
				\draw[green,\espessuralinha]  (-0.7071,-0.7071)-- (-0.9659,-1.1553);
				\draw[green,\espessuralinha]  (-0.9659,-0.2588)-- (-1.4835,-0.2588);
				\draw[blue,\espessuralinha]  (-1.4835,0.2588)-- (-1.4835,-0.2588);
				\draw[green,\espessuralinha]  (-0.9659,0.2588)-- (-1.4835,0.2588);
				\draw[green,\espessuralinha]  (0.7071,-0.7071)-- (0.9659,-1.1553);
				\draw[blue,\espessuralinha]  (0.5176,-1.4142)-- (0.9659,-1.1553);
				\draw[green,\espessuralinha]  (0.9659,-0.2588)-- (1.4835,-0.2588);
				\draw[green,\espessuralinha]  (0.9659,0.2588)-- (1.4835,0.2588);
				\draw[blue,\espessuralinha]  (1.4835,-0.2588)-- (1.4835,0.2588);
				\draw[\qubitcolor, fill=\qubitfill]  (-0.2588,0.9659) circle (\qubitraio);\draw[\qubitcolor, fill=\qubitfill]  (-0.9659,0.2588) circle (\qubitraio);\draw[\qubitcolor, fill=\qubitfill] (-0.7071,0.7071)circle (\qubitraio);
				\draw[\qubitcolor, fill=\qubitfill]  (-0.9659,-0.2588) circle (\qubitraio);\draw[\qubitcolor, fill=\qubitfill]  (-0.7071,-0.7071) circle (\qubitraio);\draw[\qubitcolor, fill=\qubitfill]  (-0.2588,-0.9659) circle (\qubitraio);\draw[\qubitcolor, fill=\qubitfill]  (0.2588,-0.9659) circle (\qubitraio);\draw[\qubitcolor, fill=\qubitfill]  (0.7071,-0.7071) circle (\qubitraio);\draw[\qubitcolor, fill=\qubitfill]  (0.9659,-0.2588) circle (\qubitraio);\draw[\qubitcolor, fill=\qubitfill]  (0.9659,0.2588) circle (\qubitraio);\draw[\qubitcolor, fill=\qubitfill]  (0.7071,0.7071) circle (\qubitraio);\draw[\qubitcolor, fill=\qubitfill]  (-0.5176,1.4142) circle (\qubitraio);\draw[\qubitcolor, fill=\qubitfill]  (-0.9659,1.1553) circle (\qubitraio);\draw[\qubitcolor, fill=\qubitfill]  (-1.4835,0.2588) circle (\qubitraio);\draw[\qubitcolor, fill=\qubitfill]  (-1.4835,-0.2588) circle (\qubitraio);\draw[\qubitcolor, fill=\qubitfill]  (-0.9659,-1.1553) circle (\qubitraio);\draw[\qubitcolor, fill=\qubitfill]  (-0.5176,-1.4142) circle (\qubitraio);\draw[\qubitcolor, fill=\qubitfill]  (0.5176,-1.4142) circle (\qubitraio);\draw[\qubitcolor, fill=\qubitfill]  (0.9659,-1.1553) circle (\qubitraio);\draw[\qubitcolor, fill=\qubitfill]  (1.4835,-0.2588) circle (\qubitraio);\draw[\qubitcolor, fill=\qubitfill]  (1.4835,0.2588) circle (\qubitraio);\draw[\qubitcolor, fill=\qubitfill]  (0.9659,1.1553) circle (\qubitraio);\draw[\qubitcolor, fill=\qubitfill]  (0.5176,1.4142) circle (\qubitraio);
		\end{scope}}
		\foreach \x in {2.44948974278318 , 4.89897948556636 , 7.34846922834953 , 9.79795897113271 , 12.2474487139159 } {\begin{scope}[xshift=\x cm]  \principal \end{scope} }
		\foreach \x in {1.22474487139159 , 3.67423461417477 , 6.12372435695795 , 8.57321409974112 , 11.0227038425243 , 13.4721935853075 } {\begin{scope}[xshift=\x cm, yshift=2.1213cm]  \principal \end{scope} }
		\foreach \x in {1.22474487139159 , 3.67423461417477 , 6.12372435695795 , 8.57321409974112 , 11.0227038425243 , 13.4721935853075 } {\begin{scope}[xshift=\x cm, yshift=-2.1213cm]  \principal \end{scope} }
		\draw[\qubitcolor, fill=\qubitfill]  (6.3815,-1.1565) circle (\qubitraio); 
		\draw[\qubitcolor, fill=\qubitfill] (3.934,-1.158)  circle (\qubitraio); 	
		\draw[draw opacity=0.3,fill opacity=0.5,text opacity=0.2, line width=1.5ex]  plot[smooth, tension=0.1] coordinates {(4.64,-2.19)  (4.64,-1.86) (4.38,-1.413) (3.934,-1.158) (4.191,-0.705) (3.933,-0.26) (3.933,0.26) (4.194,0.711) (4.64,0.965) (4.382,1.41) (4.64,1.859) (4.64,2.18)};
	\end{scope}
	
	\begin{scope}[shift={(-5.4,-4.3)}]              
		\clip  (2.4495,2.1213) node (v1) {} rectangle (7.3485,-2.1213) node (v2) {};
		\def\principal{ \begin{scope}
				\draw[fill=\configgreen] (0.2588,0.9659)--(-0.2588,0.9659)--(-0.7071,0.7071)--(-0.9659,0.2588)--(-0.9659,-0.2588)--(-0.7071,-0.7071)--(-0.2588,-0.9659)--(0.2588,-0.9659)--(-0.2588,-0.9659)--(0.2588,-0.9659)--(0.7071,-0.7071)--(0.9659,-0.2588)--(0.9659,0.2588)--(0.7071,0.7071)--cycle;
				\draw[fill=\configred]  (0.9659,-0.2588)-- (1.4835,-0.2588)-- (1.4835,0.2588)-- (0.9659,0.2588)-- (0.9659,-0.2588);
				\draw[fill=\configred]  (0.7071,0.7071)-- (0.9659,1.1553)-- (0.5176,1.4142)-- (0.2588,0.9659)-- (0.7071,0.7071);
				\draw[fill=\configred]  (-0.2588,0.9659)-- (-0.5176,1.4142)-- (-0.9659,1.1553)-- (-0.7071,0.7071)-- (-0.2588,0.9659);
				\draw[fill=\configred]  (-0.9659,0.2588)-- (-1.4835,0.2588)-- (-1.4835,-0.2588)-- (-0.9659,-0.2588)-- (-0.9659,0.2588);
				\draw[fill=\configred]  (-0.7071,-0.7071)-- (-0.9659,-1.1553)-- (-0.5176,-1.4142)-- (-0.2588,-0.9659)-- (-0.7071,-0.7071);
				\draw[fill=\configred]  (0.2588,-0.9659)-- (0.5176,-1.4142)-- (0.9659,-1.1553)-- (0.7071,-0.7071)-- (0.2588,-0.9659);
				\draw[fill=\configred]  (0.9659,-0.2588)-- (1.4835,-0.2588)-- (1.4835,0.2588)-- (0.9659,0.2588)-- (0.9659,-0.2588);
				\draw[blue!30!, fill=\configblue]  (0.2588,0.9659)-- (0.5176,1.4142)-- (-0.5176,1.4142)-- (-0.2588,0.9659)--cycle;
				\draw[blue!30!, fill=\configblue]  (-0.7071,0.7071)-- (-0.9659,1.1553)-- (-1.4835,0.2588)-- (-0.9659,0.2588)--cycle;
				\draw[blue!30!, fill=\configblue]  (-0.9659,-0.2588)-- (-1.4835,-0.2588)-- (-0.9659,-1.1553)-- (-0.7071,-0.7071)--cycle;
				\draw[blue!30!, fill=\configblue]  (-0.2588,-0.9659)-- (-0.5176,-1.4142)-- (0.5176,-1.4142)-- (0.2588,-0.9659)--cycle;
				\draw[blue!30!, fill=\configblue]  (0.7071,-0.7071)-- (0.9659,-1.1553)-- (1.4835,-0.2588)-- (0.9659,-0.2588)--cycle;
				\draw[blue!30!, fill=\configblue]  (0.9659,0.2588)-- (1.4835,0.2588)-- (0.9659,1.1553)-- (0.7071,0.7071)--cycle;
				\draw[red,\espessuralinha]  (-0.7071,0.7071)-- (-0.9659,0.2588);
				\draw[red, \espessuralinha]  (-0.9659,-0.2588)-- (-0.7071,-0.7071);
				\draw[red,\espessuralinha]  (-0.2588,-0.9659)-- (0.2588,-0.9659);
				\draw[red,\espessuralinha]  (0.7071,-0.7071)-- (0.9659,-0.2588);
				\draw[red,\espessuralinha]  (0.9659,0.2588)-- (0.7071,0.7071);
				\draw[red,\espessuralinha]  (0.2588,0.9659)-- (-0.2588,0.9659);
				\draw[blue,\espessuralinha]  (-0.9659,0.2588)-- (-0.9659,-0.2588);
				\draw[blue,\espessuralinha]  (-0.7071,-0.7071)-- (-0.2588,-0.9659);
				\draw[blue,\espessuralinha]  (0.2588,-0.9659)-- (0.7071,-0.7071);
				\draw[blue,\espessuralinha]  (0.9659,-0.2588)-- (0.9659,0.2588);
				\draw[blue,\espessuralinha]  (0.7071,0.7071)-- (0.2588,0.9659);
				\draw[blue,\espessuralinha]  (-0.2588,0.9659)-- (-0.7071,0.7071);
				\draw[fill=\configred]  (0.9659,-0.2588)-- (1.4835,-0.2588)-- (1.4835,0.2588)-- (0.9659,0.2588)-- (0.9659,-0.2588);
				\draw[green,\espessuralinha]  (0.7071,0.7071)-- (0.9659,1.1553);
				\draw[green,\espessuralinha]  (0.2588,0.9659)-- (0.5176,1.4142);
				\draw[blue,\espessuralinha]  (0.9659,1.1553)-- (0.5176,1.4142);
				\draw[green,\espessuralinha]  (-0.2588,0.9659)-- (-0.5176,1.4142);
				\draw[green,\espessuralinha]  (-0.7071,0.7071)-- (-0.9659,1.1553);
				\draw[blue,\espessuralinha]  (-0.5176,1.4142)-- (-0.9659,1.1553);
				\draw[green,\espessuralinha]  (0.2588,-0.9659)-- (0.5176,-1.4142);
				\draw[green,\espessuralinha]  (-0.2588,-0.9659)-- (-0.5176,-1.4142);
				\draw[blue,\espessuralinha]  (-0.9659,-1.1553)-- (-0.5176,-1.4142);
				\draw[green,\espessuralinha]  (-0.7071,-0.7071)-- (-0.9659,-1.1553);
				\draw[green,\espessuralinha]  (-0.9659,-0.2588)-- (-1.4835,-0.2588);
				\draw[blue,\espessuralinha]  (-1.4835,0.2588)-- (-1.4835,-0.2588);
				\draw[green,\espessuralinha]  (-0.9659,0.2588)-- (-1.4835,0.2588);
				\draw[green,\espessuralinha]  (0.7071,-0.7071)-- (0.9659,-1.1553);
				\draw[blue,\espessuralinha]  (0.5176,-1.4142)-- (0.9659,-1.1553);
				\draw[green,\espessuralinha]  (0.9659,-0.2588)-- (1.4835,-0.2588);
				\draw[green,\espessuralinha]  (0.9659,0.2588)-- (1.4835,0.2588);
				\draw[blue,\espessuralinha]  (1.4835,-0.2588)-- (1.4835,0.2588);
				\draw[\qubitcolor, fill=\qubitfill]  (-0.2588,0.9659) circle (\qubitraio);\draw[\qubitcolor, fill=\qubitfill]  (-0.9659,0.2588) circle (\qubitraio);\draw[\qubitcolor, fill=\qubitfill] (-0.7071,0.7071)circle (\qubitraio);
				\draw[\qubitcolor, fill=\qubitfill]  (-0.9659,-0.2588) circle (\qubitraio);\draw[\qubitcolor, fill=\qubitfill]  (-0.7071,-0.7071) circle (\qubitraio);\draw[\qubitcolor, fill=\qubitfill]  (-0.2588,-0.9659) circle (\qubitraio);\draw[\qubitcolor, fill=\qubitfill]  (0.2588,-0.9659) circle (\qubitraio);\draw[\qubitcolor, fill=\qubitfill]  (0.7071,-0.7071) circle (\qubitraio);\draw[\qubitcolor, fill=\qubitfill]  (0.9659,-0.2588) circle (\qubitraio);\draw[\qubitcolor, fill=\qubitfill]  (0.9659,0.2588) circle (\qubitraio);\draw[\qubitcolor, fill=\qubitfill]  (0.7071,0.7071) circle (\qubitraio);\draw[\qubitcolor, fill=\qubitfill]  (-0.5176,1.4142) circle (\qubitraio);\draw[\qubitcolor, fill=\qubitfill]  (-0.9659,1.1553) circle (\qubitraio);\draw[\qubitcolor, fill=\qubitfill]  (-1.4835,0.2588) circle (\qubitraio);\draw[\qubitcolor, fill=\qubitfill]  (-1.4835,-0.2588) circle (\qubitraio);\draw[\qubitcolor, fill=\qubitfill]  (-0.9659,-1.1553) circle (\qubitraio);\draw[\qubitcolor, fill=\qubitfill]  (-0.5176,-1.4142) circle (\qubitraio);\draw[\qubitcolor, fill=\qubitfill]  (0.5176,-1.4142) circle (\qubitraio);\draw[\qubitcolor, fill=\qubitfill]  (0.9659,-1.1553) circle (\qubitraio);\draw[\qubitcolor, fill=\qubitfill]  (1.4835,-0.2588) circle (\qubitraio);\draw[\qubitcolor, fill=\qubitfill]  (1.4835,0.2588) circle (\qubitraio);\draw[\qubitcolor, fill=\qubitfill]  (0.9659,1.1553) circle (\qubitraio);\draw[\qubitcolor, fill=\qubitfill]  (0.5176,1.4142) circle (\qubitraio);
		\end{scope}}
		\foreach \x in {2.44948974278318 , 4.89897948556636 , 7.34846922834953 , 9.79795897113271 , 12.2474487139159 } {\begin{scope}[xshift=\x cm]  \principal \end{scope} }
		\foreach \x in {1.22474487139159 , 3.67423461417477 , 6.12372435695795 , 8.57321409974112 , 11.0227038425243 , 13.4721935853075 } {\begin{scope}[xshift=\x cm, yshift=2.1213cm]  \principal \end{scope} }
		\foreach \x in {1.22474487139159 , 3.67423461417477 , 6.12372435695795 , 8.57321409974112 , 11.0227038425243 , 13.4721935853075 } {\begin{scope}[xshift=\x cm, yshift=-2.1213cm]  \principal \end{scope} }
		\draw[\qubitcolor, fill=\qubitfill]  (6.3815,-1.1565) circle (\qubitraio); 
		\draw[\qubitcolor, fill=\qubitfill] (3.934,-1.158)  circle (\qubitraio); 
		\draw[draw opacity=0.3,fill opacity=0.5,text opacity=0.2, line width=1.5ex]  plot[smooth, tension=0.1] coordinates {(4.64,-2.19)  (4.64,-1.86) (4.38,-1.413) (3.934,-1.158) (4.191,-0.705) (3.933,-0.26) (3.933,0.26) (4.194,0.711) (4.64,0.965) (4.382,1.41) (4.64,1.859) (4.64,2.18)};
	\end{scope}
	
	\begin{scope}[shift={(-0.446,-4.3)}]              
		\clip  (2.4495,2.1213) rectangle (7.3485,-2.1213);
		\def\principal{ \begin{scope}
				\draw[fill=\configgreen] (0.2588,0.9659)--(-0.2588,0.9659)--(-0.7071,0.7071)--(-0.9659,0.2588)--(-0.9659,-0.2588)--(-0.7071,-0.7071)--(-0.2588,-0.9659)--(0.2588,-0.9659)--(-0.2588,-0.9659)--(0.2588,-0.9659)--(0.7071,-0.7071)--(0.9659,-0.2588)--(0.9659,0.2588)--(0.7071,0.7071)--cycle;
				\draw[fill=\configred]  (0.9659,-0.2588)-- (1.4835,-0.2588)-- (1.4835,0.2588)-- (0.9659,0.2588)-- (0.9659,-0.2588);
				\draw[fill=\configred]  (0.7071,0.7071)-- (0.9659,1.1553)-- (0.5176,1.4142)-- (0.2588,0.9659)-- (0.7071,0.7071);
				\draw[fill=\configred]  (-0.2588,0.9659)-- (-0.5176,1.4142)-- (-0.9659,1.1553)-- (-0.7071,0.7071)-- (-0.2588,0.9659);
				\draw[fill=\configred]  (-0.9659,0.2588)-- (-1.4835,0.2588)-- (-1.4835,-0.2588)-- (-0.9659,-0.2588)-- (-0.9659,0.2588);
				\draw[fill=\configred]  (-0.7071,-0.7071)-- (-0.9659,-1.1553)-- (-0.5176,-1.4142)-- (-0.2588,-0.9659)-- (-0.7071,-0.7071);
				\draw[fill=\configred]  (0.2588,-0.9659)-- (0.5176,-1.4142)-- (0.9659,-1.1553)-- (0.7071,-0.7071)-- (0.2588,-0.9659);
				\draw[fill=\configred]  (0.9659,-0.2588)-- (1.4835,-0.2588)-- (1.4835,0.2588)-- (0.9659,0.2588)-- (0.9659,-0.2588);
				\draw[blue!30!, fill=\configblue]  (0.2588,0.9659)-- (0.5176,1.4142)-- (-0.5176,1.4142)-- (-0.2588,0.9659)--cycle;
				\draw[blue!30!, fill=\configblue]  (-0.7071,0.7071)-- (-0.9659,1.1553)-- (-1.4835,0.2588)-- (-0.9659,0.2588)--cycle;
				\draw[blue!30!, fill=\configblue]  (-0.9659,-0.2588)-- (-1.4835,-0.2588)-- (-0.9659,-1.1553)-- (-0.7071,-0.7071)--cycle;
				\draw[blue!30!, fill=\configblue]  (-0.2588,-0.9659)-- (-0.5176,-1.4142)-- (0.5176,-1.4142)-- (0.2588,-0.9659)--cycle;
				\draw[blue!30!, fill=\configblue]  (0.7071,-0.7071)-- (0.9659,-1.1553)-- (1.4835,-0.2588)-- (0.9659,-0.2588)--cycle;
				\draw[blue!30!, fill=\configblue]  (0.9659,0.2588)-- (1.4835,0.2588)-- (0.9659,1.1553)-- (0.7071,0.7071)--cycle;
				\draw[red,\espessuralinha]  (-0.7071,0.7071)-- (-0.9659,0.2588);
				\draw[red, \espessuralinha]  (-0.9659,-0.2588)-- (-0.7071,-0.7071);
				\draw[red,\espessuralinha]  (-0.2588,-0.9659)-- (0.2588,-0.9659);
				\draw[red,\espessuralinha]  (0.7071,-0.7071)-- (0.9659,-0.2588);
				\draw[red,\espessuralinha]  (0.9659,0.2588)-- (0.7071,0.7071);
				\draw[red,\espessuralinha]  (0.2588,0.9659)-- (-0.2588,0.9659);
				\draw[blue,\espessuralinha]  (-0.9659,0.2588)-- (-0.9659,-0.2588);
				\draw[blue,\espessuralinha]  (-0.7071,-0.7071)-- (-0.2588,-0.9659);
				\draw[blue,\espessuralinha]  (0.2588,-0.9659)-- (0.7071,-0.7071);
				\draw[blue,\espessuralinha]  (0.9659,-0.2588)-- (0.9659,0.2588);
				\draw[blue,\espessuralinha]  (0.7071,0.7071)-- (0.2588,0.9659);
				\draw[blue,\espessuralinha]  (-0.2588,0.9659)-- (-0.7071,0.7071);
				\draw[fill=\configred]  (0.9659,-0.2588)-- (1.4835,-0.2588)-- (1.4835,0.2588)-- (0.9659,0.2588)-- (0.9659,-0.2588);
				\draw[green,\espessuralinha]  (0.7071,0.7071)-- (0.9659,1.1553);
				\draw[green,\espessuralinha]  (0.2588,0.9659)-- (0.5176,1.4142);
				\draw[blue,\espessuralinha]  (0.9659,1.1553)-- (0.5176,1.4142);
				\draw[green,\espessuralinha]  (-0.2588,0.9659)-- (-0.5176,1.4142);
				\draw[green,\espessuralinha]  (-0.7071,0.7071)-- (-0.9659,1.1553);
				\draw[blue,\espessuralinha]  (-0.5176,1.4142)-- (-0.9659,1.1553);
				\draw[green,\espessuralinha]  (0.2588,-0.9659)-- (0.5176,-1.4142);
				\draw[green,\espessuralinha]  (-0.2588,-0.9659)-- (-0.5176,-1.4142);
				\draw[blue,\espessuralinha]  (-0.9659,-1.1553)-- (-0.5176,-1.4142);
				\draw[green,\espessuralinha]  (-0.7071,-0.7071)-- (-0.9659,-1.1553);
				\draw[green,\espessuralinha]  (-0.9659,-0.2588)-- (-1.4835,-0.2588);
				\draw[blue,\espessuralinha]  (-1.4835,0.2588)-- (-1.4835,-0.2588);
				\draw[green,\espessuralinha]  (-0.9659,0.2588)-- (-1.4835,0.2588);
				\draw[green,\espessuralinha]  (0.7071,-0.7071)-- (0.9659,-1.1553);
				\draw[blue,\espessuralinha]  (0.5176,-1.4142)-- (0.9659,-1.1553);
				\draw[green,\espessuralinha]  (0.9659,-0.2588)-- (1.4835,-0.2588);
				\draw[green,\espessuralinha]  (0.9659,0.2588)-- (1.4835,0.2588);
				\draw[blue,\espessuralinha]  (1.4835,-0.2588)-- (1.4835,0.2588);
				\draw[\qubitcolor, fill=\qubitfill]  (-0.2588,0.9659) circle (\qubitraio);\draw[\qubitcolor, fill=\qubitfill]  (-0.9659,0.2588) circle (\qubitraio);\draw[\qubitcolor, fill=\qubitfill] (-0.7071,0.7071)circle (\qubitraio);
				\draw[\qubitcolor, fill=\qubitfill]  (-0.9659,-0.2588) circle (\qubitraio);\draw[\qubitcolor, fill=\qubitfill]  (-0.7071,-0.7071) circle (\qubitraio);\draw[\qubitcolor, fill=\qubitfill]  (-0.2588,-0.9659) circle (\qubitraio);\draw[\qubitcolor, fill=\qubitfill]  (0.2588,-0.9659) circle (\qubitraio);\draw[\qubitcolor, fill=\qubitfill]  (0.7071,-0.7071) circle (\qubitraio);\draw[\qubitcolor, fill=\qubitfill]  (0.9659,-0.2588) circle (\qubitraio);\draw[\qubitcolor, fill=\qubitfill]  (0.9659,0.2588) circle (\qubitraio);\draw[\qubitcolor, fill=\qubitfill]  (0.7071,0.7071) circle (\qubitraio);\draw[\qubitcolor, fill=\qubitfill]  (-0.5176,1.4142) circle (\qubitraio);\draw[\qubitcolor, fill=\qubitfill]  (-0.9659,1.1553) circle (\qubitraio);\draw[\qubitcolor, fill=\qubitfill]  (-1.4835,0.2588) circle (\qubitraio);\draw[\qubitcolor, fill=\qubitfill]  (-1.4835,-0.2588) circle (\qubitraio);\draw[\qubitcolor, fill=\qubitfill]  (-0.9659,-1.1553) circle (\qubitraio);\draw[\qubitcolor, fill=\qubitfill]  (-0.5176,-1.4142) circle (\qubitraio);\draw[\qubitcolor, fill=\qubitfill]  (0.5176,-1.4142) circle (\qubitraio);\draw[\qubitcolor, fill=\qubitfill]  (0.9659,-1.1553) circle (\qubitraio);\draw[\qubitcolor, fill=\qubitfill]  (1.4835,-0.2588) circle (\qubitraio);\draw[\qubitcolor, fill=\qubitfill]  (1.4835,0.2588) circle (\qubitraio);\draw[\qubitcolor, fill=\qubitfill]  (0.9659,1.1553) circle (\qubitraio);\draw[\qubitcolor, fill=\qubitfill]  (0.5176,1.4142) circle (\qubitraio);
		\end{scope}}
		\foreach \x in {2.44948974278318 , 4.89897948556636 , 7.34846922834953 , 9.79795897113271 , 12.2474487139159 } {\begin{scope}[xshift=\x cm]  \principal \end{scope} }
		\foreach \x in {1.22474487139159 , 3.67423461417477 , 6.12372435695795 , 8.57321409974112 , 11.0227038425243 , 13.4721935853075 } {\begin{scope}[xshift=\x cm, yshift=2.1213cm]  \principal \end{scope} }
		\foreach \x in {1.22474487139159 , 3.67423461417477 , 6.12372435695795 , 8.57321409974112 , 11.0227038425243 , 13.4721935853075 } {\begin{scope}[xshift=\x cm, yshift=-2.1213cm]  \principal \end{scope} }
		
		\draw[\qubitcolor, fill=\qubitfill]  (6.3815,-1.1565) circle (\qubitraio); 
		\draw[\qubitcolor, fill=\qubitfill] (3.934,-1.158)  circle (\qubitraio); 
		\draw[draw opacity=0.3,fill opacity=0.5,text opacity=0.2, line width=1.5ex]  plot[smooth, tension=0.1] coordinates {(4.64,-2.19)  (4.64,-1.86) (4.38,-1.413) (3.934,-1.158) (4.191,-0.705) (3.933,-0.26) (3.933,0.26) (4.194,0.711) (4.64,0.965) (4.382,1.41) (4.64,1.859) (4.64,2.18)};
	\end{scope}
	
	\begin{scope}[shift={(4.5066,-4.3)}]              
		\clip  (2.4495,2.1213) rectangle (7.3485,-2.1213) node (v3) {};
		\def\principal{ \begin{scope}
				\draw[fill=\configgreen] (0.2588,0.9659)--(-0.2588,0.9659)--(-0.7071,0.7071)--(-0.9659,0.2588)--(-0.9659,-0.2588)--(-0.7071,-0.7071)--(-0.2588,-0.9659)--(0.2588,-0.9659)--(-0.2588,-0.9659)--(0.2588,-0.9659)--(0.7071,-0.7071)--(0.9659,-0.2588)--(0.9659,0.2588)--(0.7071,0.7071)--cycle;
				\draw[fill=\configred]  (0.9659,-0.2588)-- (1.4835,-0.2588)-- (1.4835,0.2588)-- (0.9659,0.2588)-- (0.9659,-0.2588);
				\draw[fill=\configred]  (0.7071,0.7071)-- (0.9659,1.1553)-- (0.5176,1.4142)-- (0.2588,0.9659)-- (0.7071,0.7071);
				\draw[fill=\configred]  (-0.2588,0.9659)-- (-0.5176,1.4142)-- (-0.9659,1.1553)-- (-0.7071,0.7071)-- (-0.2588,0.9659);
				\draw[fill=\configred]  (-0.9659,0.2588)-- (-1.4835,0.2588)-- (-1.4835,-0.2588)-- (-0.9659,-0.2588)-- (-0.9659,0.2588);
				\draw[fill=\configred]  (-0.7071,-0.7071)-- (-0.9659,-1.1553)-- (-0.5176,-1.4142)-- (-0.2588,-0.9659)-- (-0.7071,-0.7071);
				\draw[fill=\configred]  (0.2588,-0.9659)-- (0.5176,-1.4142)-- (0.9659,-1.1553)-- (0.7071,-0.7071)-- (0.2588,-0.9659);
				\draw[fill=\configred]  (0.9659,-0.2588)-- (1.4835,-0.2588)-- (1.4835,0.2588)-- (0.9659,0.2588)-- (0.9659,-0.2588);
				\draw[blue!30!, fill=\configblue]  (0.2588,0.9659)-- (0.5176,1.4142)-- (-0.5176,1.4142)-- (-0.2588,0.9659)--cycle;
				\draw[blue!30!, fill=\configblue]  (-0.7071,0.7071)-- (-0.9659,1.1553)-- (-1.4835,0.2588)-- (-0.9659,0.2588)--cycle;
				\draw[blue!30!, fill=\configblue]  (-0.9659,-0.2588)-- (-1.4835,-0.2588)-- (-0.9659,-1.1553)-- (-0.7071,-0.7071)--cycle;
				\draw[blue!30!, fill=\configblue]  (-0.2588,-0.9659)-- (-0.5176,-1.4142)-- (0.5176,-1.4142)-- (0.2588,-0.9659)--cycle;
				\draw[blue!30!, fill=\configblue]  (0.7071,-0.7071)-- (0.9659,-1.1553)-- (1.4835,-0.2588)-- (0.9659,-0.2588)--cycle;
				\draw[blue!30!, fill=\configblue]  (0.9659,0.2588)-- (1.4835,0.2588)-- (0.9659,1.1553)-- (0.7071,0.7071)--cycle;
				\draw[red,\espessuralinha]  (-0.7071,0.7071)-- (-0.9659,0.2588);
				\draw[red, \espessuralinha]  (-0.9659,-0.2588)-- (-0.7071,-0.7071);
				\draw[red,\espessuralinha]  (-0.2588,-0.9659)-- (0.2588,-0.9659);
				\draw[red,\espessuralinha]  (0.7071,-0.7071)-- (0.9659,-0.2588);
				\draw[red,\espessuralinha]  (0.9659,0.2588)-- (0.7071,0.7071);
				\draw[red,\espessuralinha]  (0.2588,0.9659)-- (-0.2588,0.9659);
				\draw[blue,\espessuralinha]  (-0.9659,0.2588)-- (-0.9659,-0.2588);
				\draw[blue,\espessuralinha]  (-0.7071,-0.7071)-- (-0.2588,-0.9659);
				\draw[blue,\espessuralinha]  (0.2588,-0.9659)-- (0.7071,-0.7071);
				\draw[blue,\espessuralinha]  (0.9659,-0.2588)-- (0.9659,0.2588);
				\draw[blue,\espessuralinha]  (0.7071,0.7071)-- (0.2588,0.9659);
				\draw[blue,\espessuralinha]  (-0.2588,0.9659)-- (-0.7071,0.7071);
				\draw[fill=\configred]  (0.9659,-0.2588)-- (1.4835,-0.2588)-- (1.4835,0.2588)-- (0.9659,0.2588)-- (0.9659,-0.2588);
				\draw[green,\espessuralinha]  (0.7071,0.7071)-- (0.9659,1.1553);
				\draw[green,\espessuralinha]  (0.2588,0.9659)-- (0.5176,1.4142);
				\draw[blue,\espessuralinha]  (0.9659,1.1553)-- (0.5176,1.4142);
				\draw[green,\espessuralinha]  (-0.2588,0.9659)-- (-0.5176,1.4142);
				\draw[green,\espessuralinha]  (-0.7071,0.7071)-- (-0.9659,1.1553);
				\draw[blue,\espessuralinha]  (-0.5176,1.4142)-- (-0.9659,1.1553);
				\draw[green,\espessuralinha]  (0.2588,-0.9659)-- (0.5176,-1.4142);
				\draw[green,\espessuralinha]  (-0.2588,-0.9659)-- (-0.5176,-1.4142);
				\draw[blue,\espessuralinha]  (-0.9659,-1.1553)-- (-0.5176,-1.4142);
				\draw[green,\espessuralinha]  (-0.7071,-0.7071)-- (-0.9659,-1.1553);
				\draw[green,\espessuralinha]  (-0.9659,-0.2588)-- (-1.4835,-0.2588);
				\draw[blue,\espessuralinha]  (-1.4835,0.2588)-- (-1.4835,-0.2588);
				\draw[green,\espessuralinha]  (-0.9659,0.2588)-- (-1.4835,0.2588);
				\draw[green,\espessuralinha]  (0.7071,-0.7071)-- (0.9659,-1.1553);
				\draw[blue,\espessuralinha]  (0.5176,-1.4142)-- (0.9659,-1.1553);
				\draw[green,\espessuralinha]  (0.9659,-0.2588)-- (1.4835,-0.2588);
				\draw[green,\espessuralinha]  (0.9659,0.2588)-- (1.4835,0.2588);
				\draw[blue,\espessuralinha]  (1.4835,-0.2588)-- (1.4835,0.2588);
				\draw[\qubitcolor, fill=\qubitfill]  (-0.2588,0.9659) circle (\qubitraio);\draw[\qubitcolor, fill=\qubitfill]  (-0.9659,0.2588) circle (\qubitraio);\draw[\qubitcolor, fill=\qubitfill] (-0.7071,0.7071)circle (\qubitraio);
				\draw[\qubitcolor, fill=\qubitfill]  (-0.9659,-0.2588) circle (\qubitraio);\draw[\qubitcolor, fill=\qubitfill]  (-0.7071,-0.7071) circle (\qubitraio);\draw[\qubitcolor, fill=\qubitfill]  (-0.2588,-0.9659) circle (\qubitraio);\draw[\qubitcolor, fill=\qubitfill]  (0.2588,-0.9659) circle (\qubitraio);\draw[\qubitcolor, fill=\qubitfill]  (0.7071,-0.7071) circle (\qubitraio);\draw[\qubitcolor, fill=\qubitfill]  (0.9659,-0.2588) circle (\qubitraio);\draw[\qubitcolor, fill=\qubitfill]  (0.9659,0.2588) circle (\qubitraio);\draw[\qubitcolor, fill=\qubitfill]  (0.7071,0.7071) circle (\qubitraio);\draw[\qubitcolor, fill=\qubitfill]  (-0.5176,1.4142) circle (\qubitraio);\draw[\qubitcolor, fill=\qubitfill]  (-0.9659,1.1553) circle (\qubitraio);\draw[\qubitcolor, fill=\qubitfill]  (-1.4835,0.2588) circle (\qubitraio);\draw[\qubitcolor, fill=\qubitfill]  (-1.4835,-0.2588) circle (\qubitraio);\draw[\qubitcolor, fill=\qubitfill]  (-0.9659,-1.1553) circle (\qubitraio);\draw[\qubitcolor, fill=\qubitfill]  (-0.5176,-1.4142) circle (\qubitraio);\draw[\qubitcolor, fill=\qubitfill]  (0.5176,-1.4142) circle (\qubitraio);\draw[\qubitcolor, fill=\qubitfill]  (0.9659,-1.1553) circle (\qubitraio);\draw[\qubitcolor, fill=\qubitfill]  (1.4835,-0.2588) circle (\qubitraio);\draw[\qubitcolor, fill=\qubitfill]  (1.4835,0.2588) circle (\qubitraio);\draw[\qubitcolor, fill=\qubitfill]  (0.9659,1.1553) circle (\qubitraio);\draw[\qubitcolor, fill=\qubitfill]  (0.5176,1.4142) circle (\qubitraio);
		\end{scope}}
		\foreach \x in {2.44948974278318 , 4.89897948556636 , 7.34846922834953 , 9.79795897113271 , 12.2474487139159 } {\begin{scope}[xshift=\x cm]  \principal \end{scope} }
		\foreach \x in {1.22474487139159 , 3.67423461417477 , 6.12372435695795 , 8.57321409974112 , 11.0227038425243 , 13.4721935853075 } {\begin{scope}[xshift=\x cm, yshift=2.1213cm]  \principal \end{scope} }
		\foreach \x in {1.22474487139159 , 3.67423461417477 , 6.12372435695795 , 8.57321409974112 , 11.0227038425243 , 13.4721935853075 } {\begin{scope}[xshift=\x cm, yshift=-2.1213cm]  \principal \end{scope} }
		
		\draw[\qubitcolor, fill=\qubitfill]  (6.3815,-1.1565) circle (\qubitraio); 
		\draw[\qubitcolor, fill=\qubitfill] (3.934,-1.158)  circle (\qubitraio); 
		\draw[draw opacity=0.3,fill opacity=0.5,text opacity=0.2, line width=1.5ex]  plot[smooth, tension=0.1] coordinates {(4.64,-2.19)  (4.64,-1.86) (4.38,-1.413) (3.934,-1.158) (4.191,-0.705) (3.933,-0.26) (3.933,0.26) (4.194,0.711) (4.64,0.965) (4.382,1.41) (4.64,1.859) (4.64,2.18)};
		
	\end{scope}
	
	\def\distlabel{0}
	
	\begin{scope}[shift={(-5.4279,0)}]
		\node[left=\distlabel cm] at (4.64,1.859) {$Y$};
		\node[left=\distlabel cm] at (4.3815,1.4105) {$Y$};
		\node[left=\distlabel cm] at (4.64,0.965) {$$};
		\node[left=\distlabel cm] at (4.194,0.711) {$Y$};
		\node[left=\distlabel cm] at (3.933,0.26) {$Y$};
		\node[left=\distlabel cm] at (3.933,-0.26) {$Y$};
		\node[left=\distlabel cm] at (4.191,-0.705) {$Y$};
		\node[left=\distlabel cm] at (3.934,-1.158) {$$};
		\node[left=\distlabel cm] at (4.38,-1.413) {$Y$};
		\node[left=\distlabel cm] at (4.64,-1.86) {$Y$};
	\end{scope}
	
	\begin{scope}[shift={(-0.4221,0)}]
		\node[left=\distlabel cm] at (4.64,1.859) {$Y$};
		\node[left=\distlabel cm] at (4.3815,1.4105) {$$};
		\node[left=\distlabel cm] at (4.64,0.965) {$Y$};
		\node[left=\distlabel cm] at (4.194,0.711) {$Y$};
		\node[left=\distlabel cm] at (3.933,0.26) {$Y$};
		\node[left=\distlabel cm] at (3.933,-0.26) {$Y$};
		\node[left=\distlabel cm] at (4.191,-0.705) {$$};
		\node[left=\distlabel cm] at (3.934,-1.158) {$Y$};
		\node[left=\distlabel cm] at (4.38,-1.413) {$Y$};
		\node[left=\distlabel cm] at (4.64,-1.86) {$Y$};
	\end{scope}
	
	\begin{scope}[shift={(4.4955,0)}]
		\node[left=\distlabel cm] at (4.64,1.859) {$X$};
		\node[left=\distlabel cm] at (4.3815,1.4105) {$$};
		\node[left=\distlabel cm] at (4.64,0.965) {$X$};
		\node[left=\distlabel cm] at (4.194,0.711) {$X$};
		\node[left=\distlabel cm] at (3.933,0.26) {$X$};
		\node[left=\distlabel cm] at (3.933,-0.26) {$X$};
		\node[left=\distlabel cm] at (4.191,-0.705) {$$};
		\node[left=\distlabel cm] at (3.934,-1.158) {$X$};
		\node[left=\distlabel cm] at (4.38,-1.413) {$X$};
		\node[left=\distlabel cm] at (4.64,-1.86) {$X$};
	\end{scope}
	
	\begin{scope}[shift={(4.4955,-4.2957)}]
		\node[left=\distlabel cm] at (4.64,1.859) {$$};
		\node[left=\distlabel cm] at (4.3815,1.4105) {$X$};
		\node[left=\distlabel cm] at (4.64,0.965) {$X$};
		\node[left=\distlabel cm] at (4.194,0.711) {$$};
		\node[left=\distlabel cm] at (3.933,0.26) {$$};
		\node[left=\distlabel cm] at (3.933,-0.26) {$$};
		\node[left=\distlabel cm] at (4.191,-0.705) {$X$};
		\node[left=\distlabel cm] at (3.934,-1.158) {$X$};
		\node[left=\distlabel cm] at (4.38,-1.413) {$$};
		\node[left=\distlabel cm] at (4.64,-1.86) {$$};
	\end{scope}
	
	\begin{scope}[shift={(-0.4662,-4.2957)}]
		\node[left=\distlabel cm] at (4.64,1.859) {$$};
		\node[left=\distlabel cm] at (4.3815,1.4105) {$Z$};
		\node[left=\distlabel cm] at (4.64,0.965) {$Z$};
		\node[left=\distlabel cm] at (4.194,0.711) {$$};
		\node[left=\distlabel cm] at (3.933,0.26) {$$};
		\node[left=\distlabel cm] at (3.933,-0.26) {$$};
		\node[left=\distlabel cm] at (4.191,-0.705) {$Z$};
		\node[left=\distlabel cm] at (3.934,-1.158) {$Z$};
		\node[left=\distlabel cm] at (4.38,-1.413) {$$};
		\node[left=\distlabel cm] at (4.64,-1.86) {$$};
	\end{scope}
	
	\begin{scope}[shift={(-5.3946,-4.2957)}]
		\node[left=\distlabel cm] at (4.64,1.859) {$Z$};
		\node[left=\distlabel cm] at (4.3815,1.4105) {$Z$};
		\node[left=\distlabel cm] at (4.64,0.965) {$$};
		\node[left=\distlabel cm] at (4.194,0.711) {$Z$};
		\node[left=\distlabel cm] at (3.933,0.26) {$Z$};
		\node[left=\distlabel cm] at (3.933,-0.26) {$Z$};
		\node[left=\distlabel cm] at (4.191,-0.705) {$Z$};
		\node[left=\distlabel cm] at (3.934,-1.158) {$$};
		\node[left=\distlabel cm] at (4.38,-1.413) {$Z$};
		\node[left=\distlabel cm] at (4.64,-1.86) {$Z$};
	\end{scope}
	
	\node[align=center] at (2.0003,2.6536) { \tiny \color{green} {YY measurement}   \\  \color{green} { Round 6j+1}  };
	\node[align=center] at (6.8954,2.6536) {\tiny  \color{blue} {ZZ measurement} \\  \color{blue} {Round 6j+2 } };
	\node[align=center,rotate=270] at (12.4175,-2.1332) {\color{red} {XX measurement} \\ \color{red}{Round 6j+3} };
	\node[align=center] at (6.9334,-6.9469) { \tiny \color{blue}{YY measurement} \\ \color{blue}{ Round 6j+4}  };
	\node[align=center] at (1.9334,-6.9469) {  \tiny\color{green} {YY measurement}  \\ \color{green}{Round 6j+5}  };
	\node[rotate=90, align=center] at (-3.5073,-2.1) {\tiny\color{red} XX measurement \\ \color{red} {Round 6j+6} };
	
	\draw[-latex] (5.8285,2.195) -- (7.8285,2.195);
	\draw[-latex] (1,2.195) -- (3,2.195);
	\draw[-latex] (11.9847,-1.1665)--(11.9847,-3.1665);
	\draw[-latex] (-3.0438,-3.1665)--(-3.0438,-1.1665);
	\draw[-latex] (7.9995,-6.4942) -- (5.9995,-6.4942);
	\draw[-latex] (3,-6.4942) -- (1,-6.4942);
	\end{tikzpicture}
\caption{The check measurements in the rounds.}
\end{figure}
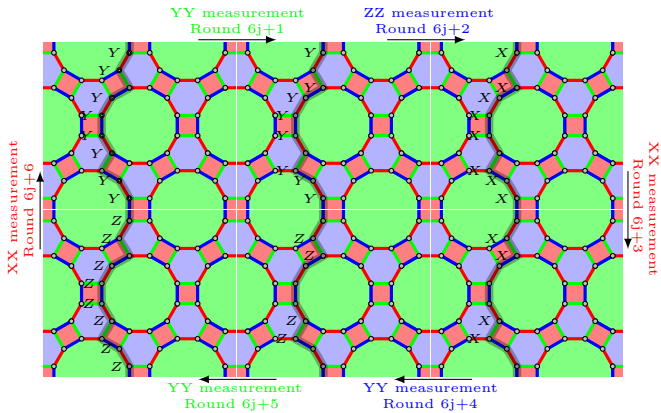

Each stabilizer generator of the Floquet code consists of a cycle of checks surrounding a face: stabilizer generators on green, blue and red faces are of X type, Y type, and Z type, respectively. However, the full stabilizer group varies at each stage. More precisely, in a given instant,  we have all cycles  and also additional stabilizers corresponding to the check operators measured in the most-recent subround. The logical operators are the operators that commute with all elements of ISG but are not in ISG. As usual, the logical operators of the new codes proposed here are contained within homollogically nontrivial closed paths of the corresponding hyperbolic tessellation utilized to construct the code.

\section{Constructions of new Floquet Codes}\label{NewFloquet}

In this section, we utilize the techniques discussed in Section \ref{sec:semireg} in order to generate trivalent and $3$-colorable tessellations in surfaces of genus $g\geq2$ and, based on this, we compute the parameters of the hyperbolic Floquet codes derived from these tessellations.

We start by presenting some examples in the case $g=2$ to illustrate the technique employed, performing in the sequence the computation of the parameters of the codes. After this procedure, we can extend the reasoning to surfaces of all genus.

Consider the $\{8,8,8\}$ hyperbolic regular tessellation which is trivalent and $3$-colorable. Observe that we can tilling the $2$-torus utilizing this tessellation, as already showed in  ~\cite{waldir2018},~\cite{albuquerque2009},~\cite{Breuckmann2016}. In this manner we obtain a tessellation with parameters $n_f =6$, $n_e=24$ and $n_v = 16$. 

The length  $n$ of the Floquet code is determined by the number $n_v$ of the corresponding tessellation, i.e., $n=n_v$. The number $k$ of the encoded qubits as occurs in the surface codes, is determined by the genus $g$ of the surface and considering the fact of orientability or not. More precisely, for a orientable surface it follows that $k=2g$ (which is the case of the $2$-torus), whereas for a non-orientable ones we have $k=g$.

\subsection{Distance of Floquet codes}

We now explain how to compute the minimum distance of the new Floquet codes. 
Although it is not necessary to utilize explicitly the concept of 
restricted tessellation \cite{Breuckmann2024}, we implicitly make use of the logical operators 
${\overline{X}}_{i}=Y_{u}Y_{v}$ (or $Z_{u}Z_{v}$) and ${\overline{Z}}_{i}=X_{u}I_{v}$ 
(or $I_{u}X_{v}$) corresponding to an edge $e=(u, v)$.

The first step is to compute the (hyperbolic) length of a geodesic $g$ of the color code 
tiling corresponding to a nontrivial homological cycle of the  Floquet code, i.e., 
the logical paths that we consider in this paper are geodesics.

We start by considering logical paths of type ${\overline{X}}_{i}$. We also consider 
the red edges contained in a given logical path of this type. We know that each 
red edge connects two corresponding red faces. Observe that some edges of the logical path belong to some edges of these red faces; see Fig. \ref{fig:logical-operators-1}.

\begin{figure}[H]
	\centering
	\label{fig:logical-operators-1}
	\begin{tikzpicture}[scale=1.5]
		\def\qubitcolor{black}
		\def\qubitfill{gray!50!}
		\def\configred{red!50!}
		\def\configblue{blue!30!}
		\def\configgreen{green!50!}
		\def\espessuralinha{very thick}
		\def\qubitraio{0.06}
		\def\qubitsuporteraio{0.03}
		\def\qubitsuportecolor{black}
		\def\qubitsuportefill{gray!30!}
		
		\clip  (2.4494897427831779,2.1213203435596424) rectangle (7.3484692283495336,-2.1213203435596424);
		\def\principal{ \begin{scope}
				\draw[fill=\configgreen] (0.2588,0.9659)--(-0.2588,0.9659)--(-0.7071,0.7071)--(-0.9659,0.2588)--(-0.9659,-0.2588)--(-0.7071,-0.7071)--(-0.2588,-0.9659)--(0.2588,-0.9659)--(-0.2588,-0.9659)--(0.2588,-0.9659)--(0.7071,-0.7071)--(0.9659,-0.2588)--(0.9659,0.2588)--(0.7071,0.7071)--cycle;
				\draw[fill=\configred]  (0.9659,-0.2588)-- (1.4835,-0.2588)-- (1.4835,0.2588)-- (0.9659,0.2588)-- (0.9659,-0.2588);
				\draw[fill=\configred]  (0.7071,0.7071)-- (0.9659,1.1553)-- (0.5176,1.4142)-- (0.2588,0.9659)-- (0.7071,0.7071);
				\draw[fill=\configred]  (-0.2588,0.9659)-- (-0.5176,1.4142)-- (-0.9659,1.1553)-- (-0.7071,0.7071)-- (-0.2588,0.9659);
				\draw[fill=\configred]  (-0.9659,0.2588)-- (-1.4835,0.2588)-- (-1.4835,-0.2588)-- (-0.9659,-0.2588)-- (-0.9659,0.2588);
				\draw[fill=\configred]  (-0.7071,-0.7071)-- (-0.9659,-1.1553)-- (-0.5176,-1.4142)-- (-0.2588,-0.9659)-- (-0.7071,-0.7071);
				\draw[fill=\configred]  (0.2588,-0.9659)-- (0.5176,-1.4142)-- (0.9659,-1.1553)-- (0.7071,-0.7071)-- (0.2588,-0.9659);
				\draw[fill=\configred]  (0.9659,-0.2588)-- (1.4835,-0.2588)-- (1.4835,0.2588)-- (0.9659,0.2588)-- (0.9659,-0.2588);
				\draw[blue!30!, fill=\configblue]  (0.2588,0.9659)-- (0.5176,1.4142)-- (-0.5176,1.4142)-- (-0.2588,0.9659)--cycle;
				\draw[blue!30!, fill=\configblue]  (-0.7071,0.7071)-- (-0.9659,1.1553)-- (-1.4835,0.2588)-- (-0.9659,0.2588)--cycle;
				\draw[blue!30!, fill=\configblue]  (-0.9659,-0.2588)-- (-1.4835,-0.2588)-- (-0.9659,-1.1553)-- (-0.7071,-0.7071)--cycle;
				\draw[blue!30!, fill=\configblue]  (-0.2588,-0.9659)-- (-0.5176,-1.4142)-- (0.5176,-1.4142)-- (0.2588,-0.9659)--cycle;
				\draw[blue!30!, fill=\configblue]  (0.7071,-0.7071)-- (0.9659,-1.1553)-- (1.4835,-0.2588)-- (0.9659,-0.2588)--cycle;
				\draw[blue!30!, fill=\configblue]  (0.9659,0.2588)-- (1.4835,0.2588)-- (0.9659,1.1553)-- (0.7071,0.7071)--cycle;
				\draw[red,\espessuralinha]  (-0.7071,0.7071)-- (-0.9659,0.2588);
				\draw[red, \espessuralinha]  (-0.9659,-0.2588)-- (-0.7071,-0.7071);
				\draw[red,\espessuralinha]  (-0.2588,-0.9659)-- (0.2588,-0.9659);
				\draw[red,\espessuralinha]  (0.7071,-0.7071)-- (0.9659,-0.2588);
				\draw[red,\espessuralinha]  (0.9659,0.2588)-- (0.7071,0.7071);
				\draw[red,\espessuralinha]  (0.2588,0.9659)-- (-0.2588,0.9659);
				\draw[blue,\espessuralinha]  (-0.9659,0.2588)-- (-0.9659,-0.2588);
				\draw[blue,\espessuralinha]  (-0.7071,-0.7071)-- (-0.2588,-0.9659);
				\draw[blue,\espessuralinha]  (0.2588,-0.9659)-- (0.7071,-0.7071);
				\draw[blue,\espessuralinha]  (0.9659,-0.2588)-- (0.9659,0.2588);
				\draw[blue,\espessuralinha]  (0.7071,0.7071)-- (0.2588,0.9659);
				\draw[blue,\espessuralinha]  (-0.2588,0.9659)-- (-0.7071,0.7071);
				\draw[fill=\configred]  (0.9659,-0.2588)-- (1.4835,-0.2588)-- (1.4835,0.2588)-- (0.9659,0.2588)-- (0.9659,-0.2588);
				\draw[green,\espessuralinha]  (0.7071,0.7071)-- (0.9659,1.1553);
				\draw[green,\espessuralinha]  (0.2588,0.9659)-- (0.5176,1.4142);
				\draw[blue,\espessuralinha]  (0.9659,1.1553)-- (0.5176,1.4142);
				\draw[green,\espessuralinha]  (-0.2588,0.9659)-- (-0.5176,1.4142);
				\draw[green,\espessuralinha]  (-0.7071,0.7071)-- (-0.9659,1.1553);
				\draw[blue,\espessuralinha]  (-0.5176,1.4142)-- (-0.9659,1.1553);
				\draw[green,\espessuralinha]  (0.2588,-0.9659)-- (0.5176,-1.4142);
				\draw[green,\espessuralinha]  (-0.2588,-0.9659)-- (-0.5176,-1.4142);
				\draw[blue,\espessuralinha]  (-0.9659,-1.1553)-- (-0.5176,-1.4142);
				\draw[green,\espessuralinha]  (-0.7071,-0.7071)-- (-0.9659,-1.1553);
				\draw[green,\espessuralinha]  (-0.9659,-0.2588)-- (-1.4835,-0.2588);
				\draw[blue,\espessuralinha]  (-1.4835,0.2588)-- (-1.4835,-0.2588);
				\draw[green,\espessuralinha]  (-0.9659,0.2588)-- (-1.4835,0.2588);
				\draw[green,\espessuralinha]  (0.7071,-0.7071)-- (0.9659,-1.1553);
				\draw[blue,\espessuralinha]  (0.5176,-1.4142)-- (0.9659,-1.1553);
				\draw[green,\espessuralinha]  (0.9659,-0.2588)-- (1.4835,-0.2588);
				\draw[green,\espessuralinha]  (0.9659,0.2588)-- (1.4835,0.2588);
				\draw[blue,\espessuralinha]  (1.4835,-0.2588)-- (1.4835,0.2588);
				\draw[\qubitcolor, fill=\qubitfill]  (-0.2588,0.9659) circle (\qubitraio);\draw[\qubitcolor, fill=\qubitfill]  (-0.9659,0.2588) circle (\qubitraio);\draw[\qubitcolor, fill=\qubitfill] (-0.7071,0.7071)circle (\qubitraio);
				\draw[\qubitcolor, fill=\qubitfill]  (-0.9659,-0.2588) circle (\qubitraio);\draw[\qubitcolor, fill=\qubitfill]  (-0.7071,-0.7071) circle (\qubitraio);\draw[\qubitcolor, fill=\qubitfill]  (-0.2588,-0.9659) circle (\qubitraio);\draw[\qubitcolor, fill=\qubitfill]  (0.2588,-0.9659) circle (\qubitraio);\draw[\qubitcolor, fill=\qubitfill]  (0.7071,-0.7071) circle (\qubitraio);\draw[\qubitcolor, fill=\qubitfill]  (0.9659,-0.2588) circle (\qubitraio);\draw[\qubitcolor, fill=\qubitfill]  (0.9659,0.2588) circle (\qubitraio);\draw[\qubitcolor, fill=\qubitfill]  (0.7071,0.7071) circle (\qubitraio);\draw[\qubitcolor, fill=\qubitfill]  (-0.5176,1.4142) circle (\qubitraio);\draw[\qubitcolor, fill=\qubitfill]  (-0.9659,1.1553) circle (\qubitraio);\draw[\qubitcolor, fill=\qubitfill]  (-1.4835,0.2588) circle (\qubitraio);\draw[\qubitcolor, fill=\qubitfill]  (-1.4835,-0.2588) circle (\qubitraio);\draw[\qubitcolor, fill=\qubitfill]  (-0.9659,-1.1553) circle (\qubitraio);\draw[\qubitcolor, fill=\qubitfill]  (-0.5176,-1.4142) circle (\qubitraio);\draw[\qubitcolor, fill=\qubitfill]  (0.5176,-1.4142) circle (\qubitraio);\draw[\qubitcolor, fill=\qubitfill]  (0.9659,-1.1553) circle (\qubitraio);\draw[\qubitcolor, fill=\qubitfill]  (1.4835,-0.2588) circle (\qubitraio);\draw[\qubitcolor, fill=\qubitfill]  (1.4835,0.2588) circle (\qubitraio);\draw[\qubitcolor, fill=\qubitfill]  (0.9659,1.1553) circle (\qubitraio);\draw[\qubitcolor, fill=\qubitfill]  (0.5176,1.4142) circle (\qubitraio);
		\end{scope}}
		\foreach \x in {2.44948974278318 , 4.89897948556636 , 7.34846922834953 , 9.79795897113271 , 12.2474487139159 } {\begin{scope}[xshift=\x cm]  \principal \end{scope} }
		\foreach \x in {1.22474487139159 , 3.67423461417477 , 6.12372435695795 , 8.57321409974112 , 11.0227038425243 , 13.4721935853075 } {\begin{scope}[xshift=\x cm, yshift=2.12132034355964 cm]  \principal \end{scope} }
		\foreach \x in {1.22474487139159 , 3.67423461417477 , 6.12372435695795 , 8.57321409974112 , 11.0227038425243 , 13.4721935853075 } {\begin{scope}[xshift=\x cm, yshift=-2.12132034355964 cm]  \principal \end{scope} }
		
		\draw[\qubitcolor, fill=\qubitfill]  (6.3815,-1.1565) circle (\qubitraio); 
		\draw[\qubitcolor, fill=\qubitfill] (3.934,-1.158)  circle (\qubitraio); 
		
		\draw[draw opacity=0.3,fill opacity=0.5,text opacity=0.2, line width=1.5ex]  plot[smooth, tension=.1] coordinates {(4.64,-2.19)  (4.64,-1.86) (4.38,-1.413) (3.934,-1.158) (4.191,-0.705) (3.933,-0.26) (3.933,0.26) (4.194,0.711) (4.64,0.965) (4.382,1.41) (4.64,1.859) (4.64,2.18)};
		\draw[draw opacity=0.3,fill opacity=0.5,text opacity=0.2, line width=1.5ex]  plot[smooth, tension=.1] coordinates {(7.09,2.165) (7.09,1.8595) (6.8305,1.4135) (6.383,1.154) (6.642,0.708) (6.384,0.261) (6.381,-0.26) (6.6425,-0.705) (6.3815,-1.1565) (6.8335,-1.414) (7.09,-1.8605) (7.09,-2.3)};
		
		\draw[\qubitsuportecolor, fill=\qubitsuportefill]  (4.64,1.859) circle (\qubitsuporteraio);
		\draw[\qubitsuportecolor, fill=\qubitsuportefill]  (4.3815,1.4105) circle (\qubitsuporteraio);
		\draw[\qubitsuportecolor, fill=\qubitsuportefill]  (4.194,0.711) circle (\qubitsuporteraio);
		\draw[\qubitsuportecolor, fill=\qubitsuportefill]  (3.933,0.26) circle (\qubitsuporteraio);
		\draw[\qubitsuportecolor, fill=\qubitsuportefill]  (3.933,-0.26) circle (\qubitsuporteraio);
		\draw[\qubitsuportecolor, fill=\qubitsuportefill]  (4.191,-0.705) circle (\qubitsuporteraio);
		\draw[\qubitsuportecolor, fill=\qubitsuportefill]  (4.38,-1.413) circle (\qubitsuporteraio);
		\draw[\qubitsuportecolor, fill=\qubitsuportefill]  (4.64,-1.86) circle (\qubitsuporteraio);
		\draw[\qubitsuportecolor, fill=\qubitsuportefill]  (6.383,1.154) circle (\qubitsuporteraio);
		\draw[\qubitsuportecolor, fill=\qubitsuportefill]  (6.3815,-1.1565) circle (\qubitsuporteraio);
		
		\node[left=-0.1cm ]  at (4.64,1.859){$Y$};
		\node[left=-0.1cm ]  at (4.3815,1.4105){$Y$};
		\node[left=-0.1cm ]  at (4.194,0.711){$Y$};
		\node[left=-0.1cm ]  at (3.933,0.26){$Y$};
		\node[left=-0.1cm ]  at (3.933,-0.26){$Y$};
		\node[left=-0.1cm ]  at (4.191,-0.705){$Y$};
		\node[left=-0.1cm ]  at (4.38,-1.413){$Y$};
		\node[left=-0.1cm ]  at (4.64,-1.86){$Y$};
		\node[left=-0.1cm ]  at (6.383,1.154){$X$};
		\node[left=-0.1cm ]  at (6.3815,-1.1565){$X$};
		
		\node[left=-0.1cm ]  at (6.642,0.708) {$X$};
		\node[left=-0.1cm ]  at   (6.381,-0.26) {$X$};
		
	\end{tikzpicture}
	\caption{The logical operators of the two logical qubits of the
		Floquet code.}
\end{figure}
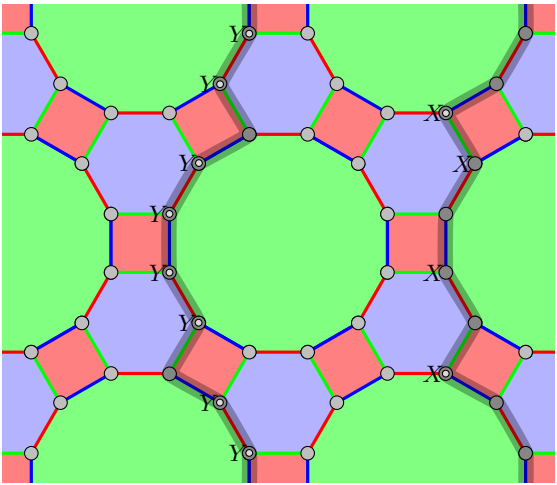

The second step is to calculate the (hyperbolic) length of the geodesic segment between the incenters of these red faces. Note that there is a bijection between each red edge $e_r$ with the corresponding geodesic segment $t_{r}$ connecting the incenters of these two red faces connected by $e_r$, i.e., the number of red edges is equal to the number of geodesic segments connecting the incenters of these two corresponding red faces. 

The third step is to divide the length of the logical path $g$ by the length of the segment $t_{r}$, after taking the ceiling function of this value, 
which results approximately in an estimate of the number of red edges contained in the logical path. Therefore, multiplying by 2, we have a lower bound for the weight of the logical operator contained in this logical path. The same procedure is adopted for all operators ${\overline{X}}_{i}$. We then take the minimum of all these weights.

In case of operators of the form ${\overline{Z}}_{i}$, we proceed in a similar manner by considering the blue and green faces as follows. Once the geodesic logical path is given, we consider only red edges which have a unique vertex belonging to the logical path. We then compute the geodesic length $t_{g, b}$ connecting the incenters of the closest blue and green faces that are adjacent to the logical path and cross such red edges. Analogously, the number of red edges that intersect the logical path in only one vertex and the number of geodesic segments $t_{g, b}$ connecting the incenters of the closest blue and green faces that cross these types of red edges are the same; see Fig. (fazer uma figura parecida com a Figura 6 do Breu sem considerar a tesselação restrita e sem considerar os inner logical). 

In this case, the weight of the logical operator is lower bounded by the ceiling function of the quotient $\displaystyle \lceil length(g)/length (t_{g, b})\rceil$.

In the same manner, we take the minimum of the weights of all operators of the form ${\overline{Z}}_{i}$. 

Finally, we take the minimum among the weights of all operators ${\overline{X}}_{i}$ and ${\overline{Z}}_{i}$, which is the code distance.

Over the $[m_1, m_2, m_3]$ semi-regular tessellation we can see that, if $A$ and $C$ are the incenters of $m_i$-gons, respectively, whose vertex is common to an edge of a $m_{i+1}$-gon whose incenter is $B$, then the hyperbolic distance between $A$ and $C$ is equal to $arccosh\left[ cosh^2(A_i)-sinh^2(A_i)cos(\dfrac{4\pi}{m_{i+1}}) \right] $, while the hyperbolic distance between $A$ and $B$ equals $A_i$, as in Proposition \ref{propo:ao-teo-kulkarni-r-valente-medidas}.

\subsection{Floquet codes over orientable surfaces}\label{subOFloquet}

Turning back to our example, the parameters of the Floquet code generated by the $[8,8,8]$ tessellation over the $2$-torus, generated by the $\{8,8\}$ polygon, identifying the opposite edges \cite{katok} are given by: $n=n_v=16$, $2=2g=4$ e $d=2$; hence, one has an $[[16,4,2]]$ code.

The parameters of such a code coincide with the ones of the code constructed in \cite{Breuckmann2024}. However, in our case, we use the techniques of derived tessellations mentioned previously to obtain semi-regular tessellations and, therefore, to construct new Floquet codes with parameters distinct from the ones displayed in the literature.

As discussed previously, if we utilize the derivation by incenter in a given $\{p,q\}$ tessellation, such a new tessellation will be of the type $\{2p,2q,4\}$, and $N_f = n_f+n_v+n_e$, $N_e = 3p n_f$ e $N_v = 2p n_f$, where $N_f$, $N_e$ and $N_v$ are the number of faces, edges and vertices of the derived tessellation, respectively. Therefore, applying such a technique to the $[8,8,8]$-tessellation (which can be denoted by the $\{8,3\}$) we obtain the semi-regular $[16,6,4]$-tessellation that generates an $[[96,4,4]]$ Floquet code.

In some cases, we can start from a tessellation which it is not trivalent nor $3$-colorable, but the derivation tessellation satisfies these two conditions. As an example, if we start with the $\{8,4\}$ tessellation we can apply the technique of Clipping, which generates an $[2p,2p,q]$-tessellation, where $N_f = n_f+n_v$, $N_e = n_e+q n_f$ and $N_v = 2 n_e$, that is, it will be generated an $\{16, 16, 4\}$ tessellation giving rise to an $[[16,4,2]]$ Floquet code. Hence, we can apply 
these two techniques to an infinity number of regular tessellations which cover some
$g$-torus, obtaininig therefore new families of Floquet codes.

In Table~\ref{tab:tess2}, we show some examples of codes generated by applying these techniques, 
separating by surface genus, and using only regular tessellations.

\begin{table}[b]
	\caption{\label{tab:tess2}
		Floquet codes generated by regular tessellations on orientable surfaces of genus $2$ to $5$.}
	\begin{ruledtabular}
		\begin{tabular}{|c|c|c|c|c|}
			$g$ & Tessellation & $[[n,k,d]]$ & $k/n$ & $kd^{2}/n$ \\
			\hline
			2 & $\{8,3\}$ & [[16,	4,	3]] & 	0.25 &	2.25 \\
			2 & $\{10,3\}$    & $[[10, 4, 2]]$  & 0.4	& 1.6	 \\
			3 & $\{8,3\}$     & $[[32, 6, 3]]$  & 0.18	& 1.68  \\
			3 & $\{10,3\}$    & $[[20, 6, 2]]$  & 0.3 &	1.2  \\
			3 & $\{14,3\}$    & $[[14, 6, 2]]$  & 0.42	& 1.71	 \\
			4 & $\{8,3\}$     & $[[48, 8, 4]]$  & 0.16	& 2.666	\\
			4 & $\{10,3\}$    & $[[30, 8, 3]]$  & 0.26	& 2.4 \\
			4 & $\{12,3\}$    & $[[24, 8, 2]]$  & 0.333	& 1.333  \\
			4 & $\{18,3\}$    & $[[18, 8, 2]]$  & 0.444	& 1.778	\\
			5 & $\{8,3\}$     & $[[64, 10, 4]]$ & 0.156	& 2.5  \\
			5 & $\{10,3\}$    & $[[40, 10, 3]]$ & 0.25	& 2.25 \\
			5 & $\{14,3\}$    & $[[28, 10, 2]]$ & 0.357	& 1.428 \\
			5 & $\{22,3\}$    & $[[22, 10, 2]]$ & 0.455	& 1.818	\\
			\end{tabular}
	\end{ruledtabular}
\end{table}

In next tables, we present the parameters of the codes obtained from derived tessellations.

\begin{table}[b]
	\caption{	\label{tab:tess3}
		Floquet codes generated by semi-regular tessellations on orientable surfaces of genus $2$.}
	\begin{ruledtabular}
		\begin{tabular}{|c|c|c|c|c|}
			$g$ & Tessellation & $[[n,k,d]]$ & $k/n$ & $kd^{2}/n$ \\
			\hline
	2 & [4, 6, 14] & [[168, 4, 5]] & 0.023 & 0.595  \\
	2 & [4, 6, 16] & [[96, 4, 4]] & 0.041 & 0.666  \\
	2 & [4, 6, 18] & [[72, 4, 3]] & 0.055 & 0.5   \\
	2 & [4, 6, 20] & [[60, 4, 3]] & 0.066 & 0.6  \\
	2 & [4, 6, 24] & [[48, 4, 3]] & 0.083 & 0.75  \\
	2 & [4, 6, 36] & [[36, 4, 2]] & 0.111 & 0.444  \\
	2 & [4, 8, 10] & [[80, 4, 4]] & 0.05 & 0.8  \\
	2 & [4, 8, 12] & [[48, 4, 3]] & 0.083 & 0.75 \\
	2 & [4, 8, 16] & [[32, 4, 3]] & 0.125 & 1.125  \\
	2 & [4, 8, 24] & [[24, 4, 2]] & 0.166 & 0.666  \\
	2 & [4, 10, 10] & [[40, 4, 3]] & 0.1 & 0.9  \\
	2 & [4, 10, 20] & [[20, 4, 2]] & 0.2 & 0.8  \\
	2 & [4, 12, 12] & [[24, 4, 2]] & 0.166 & 0.666 \\
	2 & [4, 16, 16] & [[16, 4, 2]] & 0.25 & 1  \\
	2 & [6, 6, 8] & [[48 , 4, 4]] & 0.083 & 1.333 \\
	2 & [6, 6, 10] & [[30, 4, 3]] & 0.133 & 1.2 \\
	2 & [6, 6, 12] & [[24, 4, 3]] & 0.166 & 1.5  \\
	2 & [6, 6, 18] & [[18, 4, 2]] & 0.222 & 0.888  \\
	2 & [6, 8, 8] & [[24, 4, 3]] & 0.166 & 1.5  \\
	2 & [6, 12, 12] & [[12, 4, 2]] & 0.333 & 1.333  \\
	2 & [8, 8, 8] & [[16, 4, 3]] & 0.25 & 2.25  \\
	2 & [10, 10, 10] & [[10, 4, 2]] & 0.4 & 1.6 \\  
	\end{tabular}
\end{ruledtabular}
\end{table}

\begin{table}[b]
	\caption{	\label{tab5}
		Floquet codes generated by semi-regular tessellations derived from regular tessellations over orientable surface of genus $3$.}
	\begin{ruledtabular}
		\begin{tabular}{|c|c|c|c|c|}
			$g$ & Tessellation & $[[n,k,d]]$ & $k/n$ & $kd^{2}/n$ \\
			\hline
	3 & [4 , 6 , 14 ] & [[336, 6, 7]] & 0.017 & 0.875   \\
	3 & [4 , 6 , 16 ] & [[192, 6, 5]] & 0.031 & 0.781 \\
	3 & [4 , 6 , 18 ] & [[144, 6, 4]] & 0.041 & 0.666  \\
	3 & [4 , 6 , 20 ] & [[120, 6, 4]] & 0.05 & 0.8 \\
	3 & [4 , 6 , 24 ] & [[96, 6, 3]] & 0.062 & 0.562 \\
	3 & [4 , 6 , 28 ] & [[84, 6, 3]] & 0.071 & 0.642  \\
	3 & [4 , 6 , 36 ] & [[72, 6, 3]] & 0.083 & 0.75  \\
	3 & [4 , 6 , 60 ] & [[60, 6, 2]] & 0.1 & 0.4 \\
	3 & [4 , 8 , 10 ] & [[160, 6, 5]] & 0.037 & 0.937  \\
	3 & [4 , 8 , 12 ] & [[96, 6, 4]] & 0.062 & 1  \\
	3 & [4 , 8 , 16 ] & [[64, 6, 3]] & 0.093 & 0.843  \\
	3 & [4 , 8 , 24 ] & [[48, 6, 2]] & 0.125 & 0.5  \\
	3 & [4 , 8 , 40 ] & [[40, 6, 2]] & 0.15 & 0.6 \\
	3 & [4 , 10 , 10 ] & [[80, 6, 4]] & 0.075 & 1.2  \\
	3 & [4 , 10 , 12 ] & [[60, 6, 3]] & 0.1 & 0.9  \\
	3 & [4 , 10 , 20 ] & [[40, 6, 2]] & 0.15 & 0.6  \\
	3 & [4 , 12 , 12 ] & [[48, 6, 3]] & 0.125 & 1.125  \\
	3 & [4 , 12 , 18 ] & [[36, 6, 2]] & 0.166 & 0.666  \\
	3 & [4 , 14 , 28 ] & [[28, 6, 2]] & 0.214 & 0.857 \\
	3 & [4 , 16 , 16 ] & [[32, 6, 2]] & 0.187 & 0.75  \\
	3 & [4 , 24 , 24 ] & [[24, 6, 2]] & 0.25 & 1  \\
	3 & [6 , 6 , 8 ] & [[96, 6, 5]] & 0.062 & 1.562  \\
	3 & [6 , 6 , 10 ] & [[60, 6, 4]] & 0.1 & 1.6  \\
	3 & [6 , 6 , 12 ] & [[48, 6, 3]] & 0.125 & 1.125  \\
	3 & [6 , 6 , 14 ] & [[42, 6, 3]] & 0.142 & 1.285  \\
	3 & [6 , 6 , 18 ] & [[36, 6, 3]] & 0.166 & 1.5  \\
	3 & [6 , 6 , 30 ] & [[30, 6, 2]] & 0.2 & 0.8  \\
	3 & [6 , 8 , 8 ] & [[48, 6, 4]] & 0.125 & 2 \\
	3 & [6 , 8 , 24 ] & [[24, 6, 2]] & 0.25 & 1  \\
	3 & [6 , 10 , 10 ] & [[30, 6, 3]] & 0.2 & 1.8 \\
	3 & [6 , 12 , 12 ] & [[24, 6, 2]] & 0.25 & 1  \\
	3 & [6 , 18 , 18 ] & [[18, 6, 2]] & 0.333 & 1.333  \\
	3 & [8 , 8 , 8 ] & [[32, 6, 3]] & 0.187 & 1.687  \\
	3 & [8 , 8 , 12 ] & [[24, 6, 3]] & 0.25 & 2.25  \\
	3 & [8 , 16 , 16 ] & [[16, 6, 2]] & 0.375 & 1.5  \\
	3 & [10 , 10 , 10 ] & [[20, 6, 2]] & 0.3 & 1.2 \\
	3 & [14 , 14 , 14 ] & [[14, 6, 2]] & 0.428 & 1.714 \\
	\end{tabular}
\end{ruledtabular}
\end{table}

\begin{table}[b]
	\caption{\label{tab6}
		Floquet codes generated by semi-regular tessellations derived from regular tessellations on orientable surface of genus $4$.}
	\begin{ruledtabular}
		\begin{tabular}{|c|c|c|c|c|}
			$g$ & Tessellation & $[[n,k,d]]$ & $k/n$ & $kd^{2}/n$ \\ \hline
	4 & [4 , 6 , 14 ] & [[504, 8, 8]] & 0.015 & 1.015  \\
	4 & [4 , 6 , 16 ] & [[288, 8, 6]] & 0.027 & 1  \\
	4 & [4 , 6 , 18 ] & [[216, 8, 5]] & 0.037 & 0.925  \\
	4 & [4 , 6 , 20 ] & [[180, 8, 4]] & 0.044 & 0.711  \\
	4 & [4 , 6 , 24 ] & [[144, 8, 4]] & 0.055 & 0.888  \\
	4 & [4 , 6 , 30 ] & [[120, 8, 3]] & 0.066 & 0.6  \\
	4 & [4 , 6 , 36 ] & [[108, 8, 3]] & 0.074 & 0.666  \\
	4 & [4 , 6 , 48 ] & [[96, 8, 3]] & 0.083 & 0.75  \\
	4 & [4 , 6 , 84 ] & [[84, 8, 2]] & 0.095 & 0.38  \\
	4 & [4 , 8 , 10 ] & [[240, 8, 6]] & 0.033 & 1.2  \\
	4 & [4 , 8 , 12 ] & [[144, 8, 5]] & 0.055 & 1.388  \\
	4 & [4 , 8 , 14 ] & [[112, 8, 4]] & 0.071 & 1.142  \\
	4 & [4 , 8 , 16 ] & [[96, 8, 4]] & 0.083 & 1.333  \\
	4 & [4 , 8 , 20 ] & [[80, 8, 3]] & 0.1 & 0.9   \\
	4 & [4 , 8 , 24 ] & [[72, 8, 3]] & 0.111 & 1  \\
	4 & [4 , 8 , 32 ] & [[64, 8, 2]] & 0.125 & 0.5  \\
	4 & [4 , 8 , 56 ] & [[56, 8, 2]] & 0.142 & 0.571  \\
	4 & [4 , 10 , 10 ] & [[120, 8, 4]] & 0.066 & 1.066 \\
	4 & [4 , 10 , 20 ] & [[60, 8, 3]] & 0.133 & 1.2  \\
	4 & [4 , 12 , 12 ] & [[72, 8, 3]] & 0.111 & 1   \\
	4 & [4 , 12 , 24 ] & [[48, 8, 2]] & 0.166 & 0.666   \\
	4 & [4 , 14 , 14 ] & [[56, 8, 3]] & 0.142 & 1.285  \\
	4 & [4 , 16 , 16 ] & [[48, 8, 2]] & 0.166 & 0.666   \\
	4 & [4 , 18 , 36 ] & [[36, 8, 2]] & 0.222 & 0.888  \\
	4 & [4 , 20 , 20 ] & [[40, 8, 2]] & 0.2 & 0.8  \\
	4 & [4 , 32 , 32 ] & [[32, 8, 2]] & 0.25 & 1  \\
	4 & [6 , 6 , 8 ] & [[144, 8, 6]] & 0.055 & 2  \\
	4 & [6 , 6 , 10 ] & [[90, 8, 4]] & 0.088 & 1.422  \\
	4 & [6 , 6 , 12 ] & [[72, 8, 4]] & 0.111 & 1.777  \\
	4 & [6 , 6 , 18 ] & [[54, 8, 3]] & 0.148 & 1.333  \\
	4 & [6 , 6 , 24 ] & [[48, 8, 3]] & 0.166 & 1.5  \\
	4 & [6 , 6 , 42 ] & [[42, 8, 2]] & 0.19 & 0.761  \\
	4 & [6 , 8 , 8 ] & [[72, 8, 4]] & 0.111 & 1.777   \\
	4 & [6 , 8 , 12 ] & [[48, 8, 3]] & 0.166 & 1.5  \\
	4 & [6 , 10 , 30 ] & [[30, 8, 2]] & 0.266 & 1.066   \\
	4 & [6 , 12 , 12 ] & [[36, 8, 3]] & 0.222 & 2  \\
	4 & [6 , 24 , 24 ] & [[24, 8, 2]] & 0.333 & 1.333  \\
	4 & [8 , 8 , 8 ] & [[48, 8, 4]] & 0.166 & 2.666  \\
	4 & [8 , 8 , 10 ] & [[40, 8, 3]] & 0.2 & 1.8 \\
	4 & [8 , 8 , 16 ] & [[32, 8, 2]] & 0.25 & 1  \\
	4 & [8 , 12 , 24 ] & [[24, 8, 2]] & 0.333 & 1.333   \\
	4 & [10 , 10 , 10 ] & [[30, 8, 3]] & 0.266 & 2.4  \\
	4 & [10 , 20 , 20 ] & [[20, 8, 2]] & 0.4 & 1.6  \\
	4 & [12 , 12 , 12 ] & [[24, 8, 2]] & 0.333 & 1.333   \\
	4 & [18 , 18 , 18 ] & [[18, 8, 2]] & 0.444 & 1.777 \\
			\end{tabular}
	\end{ruledtabular}
\end{table}

\begin{table}[b]
	\caption{\label{og5}
		Floquet codes generated by semi-regular tessellations derived from regular tessellations on orientable surface of genus $5$.}
	\begin{ruledtabular}
		\begin{tabular}{|c|c|c|c|c|}
			$g$ & Tessellation & $[[n,k,d]]$ & $k/n$ & $kd^{2}/n$ \\ \hline
	5 & [4, 6, 14] & [[672, 10, 9]] & 0.014 & 1.205   \\
	5 & [4, 6, 16 ] & [[384, 10, 6]] & 0.026 & 0.937  \\
	5 & [4, 6, 18] & [[288, 10, 5]] & 0.034 & 0.868  \\
	5 & [4, 6, 20] & [[240, 10, 5]] & 0.041 & 1.041   \\
	5 & [4, 6, 24] & [[192, 10, 4]] & 0.052 & 0.833  \\
	5 & [4, 6, 28] & [[168, 10, 4]] & 0.059 & 0.952   \\
	5 & [4, 6, 36] & [[144, 10, 3]] & 0.069 & 0.625   \\
	5 & [4, 6, 44] & [[132, 10, 3]] & 0.075 & 0.681  \\
	5 & [4, 6, 60] & [[120, 10, 3]] & 0.083 & 0.75   \\
	5 & [4, 6, 108] & [[108, 10, 2]] & 0.092 & 0.37  \\
	5 & [4, 8, 10] & [[320, 10, 7]] & 0.031 & 1.531  \\
	5 & [4, 8, 12] & [[192, 10, 5]] & 0.052 & 1.302  \\
	5 & [4, 8, 16] & [[128, 10, 4]] & 0.078 & 1.25  \\
	5 & [4, 8, 24] & [[96, 10, 3]] & 0.104 & 0.937  \\
	5 & [4, 8, 40] & [[80, 10, 2]] & 0.125 & 0.5  \\
	5 & [4, 8, 72] & [[72, 10, 2]] & 0.138 & 0.555  \\
	5 & [4, 10, 10] & [[160, 10, 5]] & 0.062 & 1.562  \\
	5 & [4, 10, 12] & [[120, 10, 4]] & 0.083 & 1.333  \\
	5 & [4, 10, 20] & [[80, 10, 3]] & 0.125 & 1.125  \\
	5 & [4, 10, 60] & [[60, 10, 2]] & 0.166 & 0.666  \\
	5 & [4, 12, 12] & [[96, 10, 3]] & 0.104 & 0.937  \\
	5 & [4, 12, 14] & [[84, 10, 3]] & 0.119 & 1.071  \\
	5 & [4, 12, 18] & [[72, 10, 3]] & 0.138 & 1.25  \\
	5 & [4, 12, 30] & [[60, 10, 2]] & 0.166 & 0.666  \\
	5 & [4, 14, 28] & [[56, 10, 2]] & 0.178 & 0.714  \\
	5 & [4, 16, 16] & [[64, 10, 3]] & 0.156 & 1.406  \\
	5 & [4, 16, 48] & [[48, 10, 2]] & 0.208 & 0.833  \\
	5 & [4, 22, 44] & [[44, 10, 2]] & 0.227 & 0.909  \\
	5 & [4, 24, 24] & [[48, 10, 2]] & 0.208 & 0.833  \\
	5 & [4, 40, 40] & [[40, 10, 2]] & 0.25 & 1  \\
	5 & [6, 6, 8] & [[192, 10, 6]] & 0.052 & 1.875 \\
	5 & [6, 6, 10] & [[120, 10, 5]] & 0.083 & 2.083  \\
	5 & [6, 6, 12] & [[96, 10, 4]] & 0.104 & 1.666  \\
	5 & [6, 6, 14] & [[84, 10, 4]] & 0.119 & 1.904  \\
	5 & [6, 6, 18] & [[72, 10, 3]] & 0.138 & 1.25  \\
	5 & [6, 6, 22] & [[66, 10, 3]] & 0.151 & 1.363  \\
	5 & [6, 6, 30] & [[60, 10, 3]] & 0.166 & 1.5 \\
	5 & [6, 6, 54] & [[54, 10, 2]] & 0.185 & 0.74   \\
	5 & [6, 8, 8] & [[96, 10, 4]] & 0.104 & 1.666  \\
	5 & [6, 8, 24] & [[48, 10, 2]] & 0.208 & 0.833  \\
	5 & [6, 10, 10] & [[60, 10, 3]] & 0.166 & 1.5  \\
	5 & [6, 12, 12] & [[48, 10, 3]] & 0.208 & 1.875  \\
	5 & [6, 12, 36] & [[36, 10, 2]] & 0.277 & 1.111  \\
	5 & [6, 14, 14] & [[42, 10 , 2]] & 0.238 & 0.952  \\
	5 & [6, 18, 18] & [[36, 10, 2]] & 0.277 & 1.111  \\
	5 & [6, 30, 30] & [[30, 10, 2]] & 0.333 & 1.333  \\
	5 & [8, 8, 8] & [[64, 10, 4]] & 0.156 & 2.5   \\
	5 & [8, 8, 12] & [[48, 10, 3]] & 0.208 & 1.875   \\
	5 & [8, 8, 20] & [[40, 10, 2]] & 0.25 & 1   \\
	5 & [8, 16, 16] & [[32, 10, 2]] & 0.312 & 1.25  \\
	5 & [10, 10, 10] & [[40, 10, 3]] & 0.25 & 2.25  \\
	5 & [10, 10, 30] & [[30, 10, 2]] & 0.333 & 1.333   \\
	5 & [12, 24, 24] & [[24, 10, 2]] & 0.416 & 1.666  \\
	5 & [14, 14, 14] & [[28, 10, 2]] & 0.357 & 1.428  \\
	5 & [22, 22, 22] & [[22, 10, 2]] & 0.454 & 1.818 \\ 
			\end{tabular}
\end{ruledtabular}
\end{table}
	
We next focus only on one particular tessellation and then observe the parameter 
evolution when the genus increases. Considering the $[6, 6, 8]$ tessellation, 
generated by clipping the $\{3, 8\}$ tessellation. We can observe that the parameters of the codes generated from them have a coding rate significantly better than the parameters of several Floquet codes shown in the literature.

As pointed out in \cite{Breuckmann2024}, the code generated in Table~\ref{tab7} has 
a rate of $kd^2/n$, which is better than the planar and toric honeycomb codes, which have rates of $1/6$ and $1/3$, respectively. We can also note that in some examples the rate exceeds many times such rates, as the codes are constructed over surfaces of genus $30$ and $50$. It is clear that the increase is not regular, but they present a good increase when $g$ increases.

\begin{table}[b]
	\caption{\label{tab7}
		Floquet codes generated from the $[6,6,8]$ tessellation over orientable surfaces of genus $2$ to $50$.}
	\begin{ruledtabular}
		\begin{tabular}{|c|c|c|c|c|}
			$g$ & $[[n,k,d]]$ & $k/n$ & $kd^{2}/n$ & $d/n$ \\ \hline
	2 & [[48,	4,	4]] &	0.083 &	1.333 &	0.08333 \\
	3 & [[96,	6,	5]] & 	0.062 &	1.562 &	0.052  \\
	4 &	[[144,	8,	6]] &  0.055 &	2 &	0.041 \\
	5 &	[[192,	10,	6]] &  0.052 &	1.875 &	0.031 \\
	6 & [[240,	12,	7]] &  0.05 &	2.45 &	0.029 \\
	7 & [[288,	14,	7]] &  0.048 &	2.381 &	0.024 \\
	8 & [[336,	16,	5]] &   0.095 &	2.38 &	0.029 \\
	9 & [[384,	18,	8]] &   0.046 &	3 &	0.02 \\
	10 & [[432,	20,	8]] &   0.046 &	2.962 &	0.018 \\
	11 & [[480,	22,	8]] &   0.045 &	2.933 &	0.016 \\
	12 & [[528,	24,	8]] &   0.045 &	2.909 &	0.015 \\
	13 & [[576,	26,	9]] & 	0.045 &	3.656 &	0.015 \\
	14 & [[624,	28,	9]] & 	0.044 &	3.634 &	0.014  \\
	15 & [[672,	30,	9]] &  0.044 &	3.616 &	0.013  \\
	16 & [[720,	32,	9]] &   0.044 &	3.6 &	0.012 \\
	17 & [[768,	34,	9]] &  0.044 &	3.585 &	0.011 \\
	18 & [[816,	36,	9]] &   0.044 &	3.573 &	0.011 \\
	19 & [[864,	38,	10]] &  0.043 &	4.398 &	0.011 \\
	20 & [[912,	40,	10]] &   0.043 &	4.385 &	0.01 \\
	21 & [[960,	42,	10]] & 0.043 &	4.375 &	0.01 \\
	22 & [[1008,	44,	10]] &  0.043 &	4.365 &	0.009 \\
	23 & [[1056,	46,	10]] &  0.043 &	4.356 &	0.009 \\
	24 & [[1104,	48,	10]] &  0.043 &	4.347 &	0.009 \\
	25 & [[1152,	50,	10]] &   0.043 &	4.34 &	0.008 \\
	26 & [[1200,	52,	10]] &   0.043 &	4.333 &	0.008 \\
	27 & [[1248,	54,	10]] &  0.043 &	4.326 &	0.008 \\
	28 & [[1296,	56,	10]] &   0.043 &	4.32 &	0.007 \\
	29 & [[1344,	58,	10]] &  0.043 &	4.315 &	0.007 \\
	30 & [[1392,	60,	11]] &  0.043 &	5.215 &	0.007 \\
	40 & [[1872,	80,	11]] &  0.042 &	5.17 &	0.005 \\
	50 & [[2352,	100,	12]] &  0.042 &	6.122 &	0.005 \\
			\end{tabular}
\end{ruledtabular}
\end{table}

\subsection{Floquet codes over non-orientable surfaces}\label{subNOFloquet}

In this section we utilize the derived semi-regular tessellations in order to construct new Floquet codes now applied to tessellations over non-orientable surfaces.

Constructions of codes derived from semi-regular tessellations over non-orientable surfaces were already been present in the literature; see for example \cite{Douglas1}, \cite{norientavel}. 

An important result shown in \cite{norientavel} says that if a code is generated over an orientable surface of genus $g$, then one can generate a corresponding code over a non-orientable surface of genus $2g$, having the same parameters of the first code. This fact is interesting since every code generates over an orientable surface has an analogous code over a non-orientable surface of even genus. However, a code generated over non-orientable surfaces of odd genus does not have an equivalent code in the sense that they have the same parameters over orientable surfaces. In this light, we foccus the attention in the construction of codes over non-orientable surfaces of odd genus.

\begin{theorem}\label{or_same_non_or_qcc}
	Let $S_{g}$ and $S_{h}$ be surfaces with non-orientable and orientable genus $g$ and $ h $, respectively. If $\chi(S_{g}) = \chi(S_{h})$, then the Floquet codes constructed are the same for both surfaces. 
\end{theorem}

\begin{proof}
	We know that the fundamental region of $S_{h}$ is a $4h$-gon polygon. Similarly, the fundamental region of $S_g$ is a $2g$-gon polygon. If $\chi(S_{g}) = \chi(S_{h})$, then $g=2h$. Hence, the fundamental region for both surfaces, $S_g$ and $S_h$ are the same $2g$-gonal polygon with different orientations on the boundary edges. Therefore, both the surfaces admit the same existing $[p,q,r]$ tessellation, and hence, the stabilizer operators associated with the tessellation $[p,q,r]$ for both the surfaces are the same since both are independent of the orientation of the boundary edges. Therefore, the code space and the encoding space associated with the tessellation $[p,q,r]$ on both surfaces are the same.
	
	We have that the code parameter $k$ is given by $k=4-2\chi$, therefore, $k$ is the same for both surfaces $S_g$ and $S_{h}$ if $\chi(S_g)=\chi(S_{h})$.
	
	We see that $l(p,q)$, $D(p,q)$ and $AR(p,q)$ are not depending on the genus of the surface. Therefore, if $\chi(S_g)= \chi(S_{h})$, then $l(p,q)$, $D(p,q)$ and $AR(p,q)$ are same for both surfaces $S_g$ and $S_{h}$. Now, $\chi(S_g)=  \chi(S_{h})$ implies $g=2h$, and this leads to
	\[
	d_h=2arccosh\Bigg[\frac{\cos{(\frac{\pi}{4h})}}{\sin{(\frac{\pi}{4h})}}\bigg].
	\]
	This concludes that $d_{min}$ is the same for the surface $S_{h}$.
\end{proof}

\begin{table}[b]
	\caption{\label{tab8}
		Floquet codes generated from semi-regular tessellations on non-orientable surfaces of genus $3$.}
	\begin{ruledtabular}
		\begin{tabular}{|c|c|c|c|c|}
			$g$ & Tessellation & $[[n,k,d]]$ & $k/n$ & $kd^{2}/n$ \\ \hline
	3 & [6, 6, 8] & [[24, 3, 4]] & 0.125 & 2  \\
	3 & [6, 16, 4] & [[48, 3, 4]] & 0.062 & 1  \\
	3 & [6, 6, 12] & [[12, 3, 2]] & 0.25 & 1  \\
	3 & [6, 24, 4] & [[24, 3, 2]] & 0.125 & 0.5   \\
	3 & [8, 8, 6] & [[12, 3, 4]] & 0.25 & 4  \\
	3 & [8, 12, 4] & [[24, 3, 4]] & 0.125 & 2  \\
	3 & [8, 8, 8] & [[8, 3, 2]] & 0.375 & 1.5   \\
	3 & [8, 16, 4] & [[16, 3, 2]] & 0.187 & 0.75  \\
	3 & [10, 10, 4] & [[20, 3, 4]] & 0.15 & 2.4  \\
	3 & [10, 8, 4] & [[40, 3, 4]] & 0.075 & 1.2  \\
	3 & [12, 12, 4] & [[12, 3, 2]] & 0.25 & 1  \\
	3 & [12, 8, 4] & [[24, 3, 4]] & 0.125 & 2  \\
	3 & [12, 12, 6] & [[6, 3, 2]] & 0.5 & 2  \\
	3 & [12, 12, 4] & [[12, 3, 2]] & 0.25 & 1  \\
	3 & [16, 16, 4] & [[8, 3, 2]] & 0.375 & 1.5  \\
	3 & [16, 8, 4] & [[16, 3, 2]] & 0.187 & 0.75 \\ 
				\end{tabular}
\end{ruledtabular}
\end{table}

\begin{table}[b]
	\caption{\label{nog5}
		Floquet codes generated from semi-regular tessellations on non-orientable surfaces of genus $5$.}
	\begin{ruledtabular}
		\begin{tabular}{|c|c|c|c|c|}
			$g$ & Tessellation & $[[n,k,d]]$ & $k/n$ & $kd^{2}/n$ \\ \hline
	5 & [6, 6, 8] & [[72, 5, 6]] & 0.069 & 2.5  \\
	5 & [6, 16, 4] & [[144, 5, 6]] & 0.034 & 1.25  \\
	5 & [6, 6, 12] & [[36, 5, 4]] & 0.138 & 2.222  \\
	5 & [6, 24, 4] & [[72, 5, 4]] & 0.069 & 1.111  \\
	5 & [6, 6, 24] & [[24, 5, 2]] & 0.208 & 0.833   \\
	5 & [6, 48, 4] & [[48, 5, 2]] & 0.104 & 0.416   \\
	5 & [8, 8, 6] & [[36, 5, 4]] & 0.138 & 2.222  \\
	5 & [8, 12, 4] & [[72, 5, 6]] & 0.069 & 2.5   \\
	5 & [8, 8, 8] & [[24, 5, 4]] & 0.208 & 3.333  \\
	5 & [8, 16, 4] & [[48, 5, 4]] & 0.104 & 1.666  \\
	5 & [8, 8, 10] & [[20, 5, 4]] & 0.25 & 4  \\
	5 & [8, 20, 4] & [[40, 5, 4]] & 0.125 & 2  \\
	5 & [8 , 8 , 16] & [[16, 5, 2]] & 0.312 & 1.25  \\
	5 & [8, 32, 4] & [[32, 5, 2]] & 0.156 & 0.625 \\
	5 & [10, 10, 4] & [[60, 5, 6]] & 0.083 & 3  \\
	5 & [10, 8, 4] & [[120, 5, 6 ]] & 0.041 & 1.5  \\
	5 & [12, 12, 4] & [[36, 5, 4]] & 0.138 & 2.222  \\
	5 & [12, 8, 4] & [[72, 5, 6]] & 0.069 & 2.5 \\
	5 & [12, 12, 6] & [[18, 5, 4]] & 0.277 & 4.444  \\
	5 & [12, 12, 4] & [[36, 5, 4]] & 0.138 & 2.222  \\
	5 & [12, 12, 12] & [[12, 5, 2]] & 0.416 & 1.666  \\
	5 & [12, 24, 4] & [[24, 5, 2]] & 0.208 & 0.833  \\
	5 & [14, 14, 4] & [[28, 5, 4]] & 0.178 & 2.857  \\
	5 & [14, 8, 4] & [[56, 5, 4]] & 0.089 & 1.428  \\
	5 & [16, 16, 4] & [[24, 5, 4]] & 0.208 & 3.333  \\
	5 & [16, 8, 4] & [[48, 5, 4]] & 0.104 & 1.666  \\
	5 & [20, 20, 4] & [[20, 5, 2]] & 0.25 & 1  \\
	5 & [20, 8, 4] & [[40, 5, 4]] & 0.125 & 2  \\
	5 & [20, 20, 10] & [[10, 5, 2]] & 0.5 & 2  \\
	5 & [20, 20, 4] & [[20, 5, 2]] & 0.25 & 1  \\
	5 & [24, 24, 6] & [[12, 5, 2]] & 0.416 & 1.666 \\
	5 & [24, 12, 4] & [[24, 5, 2]] & 0.208 & 0.833  \\
	5 & [32, 32, 4] & [[16, 5, 2]] & 0.312 & 1.25  \\
	5 & [32, 8, 4] & [[32, 5, 2]] & 0.156 & 0.625   \\
				\end{tabular}
\end{ruledtabular}
\end{table}

\begin{table}[b]
	\caption{\label{tab10}
		Floquet codes generated from semi-regular tessellations on non-orientable surfaces of genus $7$.}
	\begin{ruledtabular}
		\begin{tabular}{|c|c|c|c|c|}
			$g$ & Tessellation & $[[n,k,d]]$ & $k/n$ & $kd^{2}/n$ \\ \hline
	7 & $[6, 6, 8]$ & [[120,	7,	6]] &	0.058 &	2.1  \\
	7 &	$[6, 16, 4]$ & [[240,	7,	6]] & 	0.029 &	1.05  \\
	7 & $[6, 6, 12]$ & [[60,	7,	4]] & 	0.116 &	1.866  \\
	7 &	$[6, 24, 4]$ & [[120,	7,	4]] &	0.058 &	0.933 \\
	7 & $[6, 6, 16]$ & [[48,	7,	4]] &	0.145 & 2.333 \\
	7 & $[6, 32, 4]$ & [[96,	7,	4]] &	0.072 &	1.166 \\
	7 & $[6, 6, 36]$ & [[36,	7,	2]] &	0.194 &	0.777  \\
	7 & $[6, 72, 4]$ & [[72,	7,	2]] &	0.097 &	0.388  \\
	7 & $[8, 8, 6]$ & [[60,	7,	6]] &	0.116 &	4.2  \\
	7 & $[8, 12, 4]$ & [[120,	7,	6]] &	0.058 &	2.1  \\
	7 &	$[8, 8, 8]$ & [[40,	7,	4]] &	0.175 &	2.8   \\
	7 & $[8, 16, 4]$ & [[80,	7,	4]] & 	0.087 &	1.4   \\
	7 & $[8, 8, 14]$ & [[28,	7,	4]] &	0.25 &	4   \\
	7 & $[8, 28, 4]$ & [[56,	7, 	4]] &	0.125 &	2   \\
	7 &	$[8, 8, 24]$ & [[24,	7,	2]] & 	0.291 &	1.166   \\
	7 & $[8, 48, 4]$ & [[48,	7,	2]] &	0.145 &	0.583   \\
	7 &	$[10, 10, 4]$ & [[100, 7,	6]] &	0.07 &	2.52   \\
	7 &	$[10, 8, 4]$ & [[200,	7,	8]] &	0.035 &	2.24  \\
	7 & $[10, 10, 20]$ & [[20, 7,	2]] &	0.35 &	1.4   \\
	7 &	$[10, 40, 4]$ & [[40,	7,	2]] &	0.175 &	0.7  \\
	7 & $[12, 12, 4]$ & [[60,	7,	4]] &	0.116 &	1.866  \\
	7 & $[12, 8, 4]$ & [[120,	7,	6]] &	0.058 &	2.1  \\
	7 & $[12, 12, 6]$ & [[30,	7,	4]] &	0.233 &	3.733   \\
	7 & $[12, 12, 4]$ & [[60,	7,	4]] & 	0.116 &	1.866   \\
	7 & $[12, 12, 8]$ & [[24,	7,	4]] &	0.291 &	4.666  \\
	7 & $[12, 16, 4]$ & [[48,	7,	4]] & 	0.145 &	2.333   \\
	7 &	$[12, 12, 18]$ & [[18, 7,	2]] &	0.388 &	1.555   \\
	7 & $[12, 36, 4]$ & [[36,	7,	2]] &	0.194 &	0.777  \\
	7 & $[16, 16, 4]$ & [[40,	7,	4]] & 	0.175 &	2.8  \\
	7 &	$[16, 8, 4]$ & [[80,	7,	4]] &	0.087 &	1.4   \\
	7 & $[16, 16, 6]$ & [[24,	7, 	4]] &	0.291 &	4.666  \\
	7 &	$[16, 12, 4]$ & [[48,	7,	4]] &	0.145 &	2.333  \\
	7 & $[16, 16, 16]$ & [[16, 7,	2]] &	0.437 &	1.75  \\ 
	7 & $[16, 32, 4]$ & [[32,	7,	2]] &	0.218 &	0.875   \\
	7 & $[18, 18, 4]$ & [[36,	7,	4]] &	0.194 &	3.111  \\
	7 & $[18, 8, 4]$ & [[72,	7,	4]] &	0.097 &	1.555 \\
	7 & $[28, 28, 4]$ & [[28,	7,	2]] &	0.25 &	1  \\
	7 & $[28, 8, 4]$ & [[56,	7,	4]] &	0.125 &	2   \\
	7 & $[28, 28, 14]$ & [[14, 7,	2]] &	0.5 &	2   \\
	7 & $[28, 28, 4]$ & [[28,	7,	2]] &	0.25 &	1  \\
	7 & $[32, 32, 8]$ & [[16,	7,	2]] &	0.437 &	1.75  \\
	7 & $[32, 16, 4]$ & [[32,	7,	2]] &	0.218 &	0.875  \\
	7 & $[36, 36, 6]$ & [[18,	7,	2]] &	0.388 &	1.555  \\
	7 & $[36, 12, 4]$ & [[36,	7,	2]] &	0.194 &	0.777   \\
	7 & $[48, 48, 4]$ & [[24,	7,	2]] &	0.291 &	1.166  \\
	7 & $[48, 8, 4]$ & [[48,	7,	2]] &	0.145 &	0.583   
					\end{tabular}
\end{ruledtabular}
\end{table}

As was made with codes over orientable surfaces, we fix the tessellation and observe its behavior according to genus increasing (see Table~\ref{668} in the sequence)

\begin{table}[b]
	\caption{\label{668}
Floquet code generated from the $[6,6,8]$ tessellation over non-orientable surfaces of genus $3$ to $51$}
	\begin{ruledtabular}
		\begin{tabular}{|c|c|c|c|c|}
	$g$ & $[[n,k,d]]$ & $k/n$ & $kd^2/n$ & $d/n$ \\ \hline
	3 & [[24,	3,	4]] &	0.125 &	2 &	0.16666  \\
	4 & [[48,	4,	4]] &	0.083 &	1.333 &	0.08333  \\
	5 & [[72,	5,	6]] &	0.069 &	2.5 &	0.08333  \\
	6 & [[96,	6,	6]] & 	0.062 &	2.25 &	0.0625  \\
	7 & [[120,	7,	6]] &	0.058 &	2.1 &	0.05  \\
	8 & [[144,	8,	8]] &	0.055 &	3.555 &	0.05555 \\
	9 & [[168,	9,	8]] &	0.053 &	3.428 &	0.04761  \\
	10 & [[192,	10,	8]] &	0.052 &	3.333 &	0.04166  \\
	11 & [[216,	11,	8]] &	0.05 &	3.259 &	0.03703  \\
	12 & [[240,	12,	8]] &	0.05 &	3.2 &	0.03333  \\
	13 & [[264,	13,	8]] &	0.049 &	3.151 &	0.0303  \\
	14 & [[288,	14,	8]] &	0.048 &	3.111 &	0.02777 \\
	15 & [[312,	15,	8]] &	0.048 &	3.076 &	0.02564  \\
	16 & [[336,	16,	8]] &	0.047 &	3.047 &	0.0238  \\
	17 & [[360,	17,	10]] &	0.047 &	4.722 &	0.02777 \\
	18 & [[384,	18,	10]] &	0.046 &	4.687 &	0.02604  \\
	19 & [[408,	19,	10]] &	0.046 &	4.656 &	0.0245  \\
	20 & [[432,	20,	10]] &	0.046 &	4.629 &	0.02314 \\
	21 & [[456,	21,	10]] &	0.046 &	4.605 &	0.02192  \\
	22 & [[480,	22,	10]] &	0.045 &	4.583 &	0.02083  \\
	23 & [[504,	23,	10]] &	0.045 &	4.563 &	0.01984  \\
	24 & [[528,	24,	10]] &	0.045 &	4.545 &	0.01893  \\
	25 & [[552,	25,	10]] &	0.045 &	4.528 &	0.01811  \\
	26 & [[576,	26,	10]] &	0.045 &	4.513 &	0.01736  \\
	27 & [[600,	27,	10]] &	0.045 &	4.5 &	0.01666 \\
	28 & [[624,	28,	10]] &	0.044 &	4.487 &	0.01602  \\
	29 & [[648,	29,	10]] &	0.044 &	4.475 &	0.01543  \\
	30 & [[672,	30,	10]] &	0.044 &	4.464 &	0.01488  \\
	31 & [[696,	31,	10]] &	0.044 &	4.454 &	0.01436  \\
	41 & [[936,	41,	12]] & 	0.043 &	6.307 &	0.01282  \\
	51 & [[1176, 51, 12]] &	0.043 &	6.244 &	0.0102  
					\end{tabular}
\end{ruledtabular}
\end{table}

In Table \ref{668} we focus the attention in the construction of codes derived from the $[6,6,8]$ tessellation due to the criterion of vertices density, and we compare with the parameters of the codes derived from regular tessellations. Note that there exists an increasing of the rate $kd^2/n$ of this code, which indicates a possibility of good coding rates.

\begin{remark}
	Note that if we compare the parameters of the codes generated over orinetable surfaces with codes generated over non-orientable surfaces of the same genus $g$, the codes over non-orientable surfaces have better coding rate $k/n$, whereas the codes over orientable surfaces have better $kd^2/n$ rate.
\end{remark}

\section{Comparisons and asymptotic behavior}\label{CAB}

Even the asymptotic behaviour of $k/n$ is the same for the orientable and the non-orientable cases when the genus grows, we have usually other better ratios for the non-orientable case for a fixed odd genus. We can compare the codes of Table \ref{og5}, for $g = 5$, with the same genus in the orientable case, given in Table \ref{nog5}. We consider also the ratios $d/n$ and $k d^2/n$, which are very important. Maximizing of ratios $k/n$ and $d/n$, we can encode as many qubits as possible and correct as many errors as possible. We remark that always $d/n \leq 1$ and the higher the ratio, the better the code. We can see that both radios $k/n$, $k d^2/n$ and $d/n$ are always better in the non-orientable case. The same occurs in the odd genera greater than $5$, as we can see in the next table.

\begin{table*}
	\caption{		\label{tab12}
		Ratios $k/n$, $kd^2/n$, and $d/n$ for Floquet codes obtained from orientable (O) and non-orientable (NO) cases for genus $9$}
	\begin{ruledtabular}
\begin{tabular}{| c | c | c | c |c|c|c|}
	Tessellation & $k/n$ (O) & $kd^2/n $ (O) & $ d/n $ (O)& $k/n$ (NO) & $kd^2/n $ (NO) & $ d/n $ (NO) \\ \hline
			$[4 , 6 , 16 ]$ & 0.023 &	1.5 &	0.01  	& 0.026 &	1.714 &	0.023 \\
			$[4,	6,	24]$ & 0.046 &	1.171 &	0.013  & 0.053 &	0.857 &	0.023 \\
			$[4,	8,	12]$ & 0.046 &	1.687 &	0.015  &   0.053 &	1.928 &	0.035 \\
			$[4,	8,	16]$ & 0.07 &	1.757 &	0.019  &  0.08 &	1.285 &	0.035 \\
			$[4 , 10 , 10 ]$ & 0.056 & 1.406 & 0.015 &  0.064 &	2.314 &	0.042\\
			$[4,	12,	12]$ & 0.093 &	1.5 &	0.02  &  0.107 &	3.857 &	0.071 \\
			$[4,	16,	16]$ &  0.14 &	1.265 &	0.023 &  0.16 &	2.571 &	0.071 \\
			$[6,	6,	12]$ & 0.093 &	2.343 &	0.026	&  0.107 &	1.714 &	0.047 \\
			$[6,	8,	8]$ & 0.093 &	2.343 &	0.026  &  0.107 &	3.857 &	0.071 \\
			$[6,	12,	12]$ & 0.187 &	1.687 &	0.031  &  0.214 &	3.428 &	0.095 \\
			\hline
							\end{tabular}
	\end{ruledtabular}
\end{table*}
	
	\subsection{Performance analysis}\label{subPA}

	Here we analyze the parameters of the obtained codes and the encoding rates. Here we are considering quasi-Euclidean and quasi-regular tessellations.
	
	The only regular trivalent and three-colorable Euclidean tessellation is $\{6,3\} = [6,6,6]$, which does not exist in the hyperbolic plane. But, using hyperbolic semi-regular tessellations, we can consider the families $[6,6,p]$ for even $p \geq 8$, which we call here quasi-Euclidean tessellation. The encoding rate for this case is $\frac{k}{n} = \frac{g}{g-1} \frac{p - 3}{3 p}$, then taking $g, p \rightarrow \infty$ we get $\frac{k}{n} \rightarrow \frac{1}{3}$.
	
	We can see in lines 2 to 6 of Table \ref{tab13} that the closer the parameters to the Euclidean case, the greater the distances when we increase the surface genus. Regular tessellations do not allow for anything similar.
	
	The Euclidean semi-regular tessellation $[4,4,8]$ is also called a quasi-regular tessellation, and it was used to obtain some color codes in the literature. We can also see the parameters of some regular tessellations $\{2p,3\} = [2p,2p,2p]$ and some close quasi-regular ones in Table \ref{tab13}. 
	
	In general, for the tessellation $[2p,2p,2q]$, the encoding rate is $\frac{k}{n} = \frac{g}{g-1}\frac{pq - p - 2q}{pq}$, and then taking $g, p, q \rightarrow \infty$, we get $\frac{k}{n} \rightarrow 1$, which reflects the power of these families.
	
	From the tables provided in this work, we can observe that the semi-regular cases not only provide codes with excellent parameters but also offer a greater degree of freedom in selecting the polygons. This allows us to include both polygons with few edges and polygons with a large number of edges in the same tessellation.

\begin{table*}
	\caption{\label{tab13}
		Parameters of the color codes generated in the $g$-tori according to the $[2p,2p,2q]$ tessellation}
	\begin{ruledtabular}
				\begin{tabular}{|c|c|c|c|c|c|c|}
			Tessellation&$g=2$ (O) &$g=3$ (O) &$g=3$ (NO)&$g=5$ (O) & $g = 5$ (NO) & $g=10$\\ \hline
			$[6,6,8]$&$[[48,8,4]]$&$[[96,6,5]]$&$[[24,3,4]]$&$[[192,10,6]]$ & [[72, 5, 6]] & $[[432, 20, 8]]$ \\ \hline
			$[6,6,10]$ & $[[30, 4 , 3]]$ & $[[60,6,4]]$ & $-$ & $[[120, 10, 5]]$ & $-$ & $[[270, 20, 6]]$ \\ \hline
			$[6,6,12]$ & $[[24, 4, 3]]$ & $[[48, 6, 3]]$ & $[[12, 3, 2]]$ &  $[[96, 10, 4]]$ & [[36, 5, 4]] & $[[216, 20, 5]]$ \\ \hline
			$[6,6,14]$ &$-$& $[[42, 6, 3]]$ & $-$ & $[[84, 10, 4]]$ & $-$ & $-$ \\ \hline
			$[6,6,18]$ & $[[18, 4, 2]]$ & $[[36, 6, 3]]$ & $-$ & $[[72, 10, 3]]$ & $-$  & $[[162, 20, 4]]$ \\ \hline
			$[8,8,6]$ & $[[24, 4, 3]]$ & $[[48, 6, 4]]$ & $[[12, 3, 4]]$ & $[[96, 10, 4]]$ & $[[36, 5, 4]]$ & $[[216, 20, 5]]$ \\ \hline
			$[8,8,8]$ & $[[16, 4, 3]]$ & $[[32, 6, 3]]$ & $[[8, 3, 2]]$ & $[[64, 10, 4]]$ & $[[24, 5, 4]]$ & $[[144, 20, 5]]$ \\ \hline
			$[8,8,10]$ & $-$ & $-$& $-$ &$-$ &$[[20, 5, 4]]$ & $[[120, 20, 4]]$ \\ \hline
			$[8,8,12]$ & $-$& $[[24, 6, 3]]$ & $-$& $[[48, 10, 3]]$ & $-$ & $-$ \\ \hline
			$[8,8,16]$ & $-$ & $-$ & $-$ & $-$ & $[[16, 5, 2]]$ & $[[96, 20, 3]]$  \\ \hline
			$[8,8,20]$ & $-$ & $-$ & $-$& $[[40, 10, 2]]$ & $-$ & $ - $ \\ \hline
			$[10,10,6]$& $-$ & $[[30, 6, 3]]$ & $-$ & $[[60, 10, 3]]$ & $-$ & $-$ \\ \hline
			$[10,10,10]$ & $[[10,4,2]]$ & $[[20, 6, 2]]$ & $-$ & $[[40, 10, 3]]$ & $-$ & $[[90, 20, 4]]$ \\ \hline
			$[12,12,6]$ & $[[12, 4, 2]]$ & $[[24, 6, 2]]$ & $[[6, 3, 2]]$ & $[[48, 10, 3]]$ & $[[18, 5, 4]]$ & [[108, 20 , 3]] \\ \hline
			$[12,12,12]$ & $-$ &$-$& $-$ & $-$&$ [[12, 5, 2]]$ & $[[72, 30, 3]]$ \\ \hline
			$[14,14,6]$ & $-$ & $-$ & $-$ & $[[42, 10, 2]]$ & $-$ & $-$ \\ \hline
			$[14,14,10]$ & $-$ & $-$ & $-$ & $-$ & $-$ & $[[70, 20, 3]]$ \\ \hline
			$[14,14,14]$ & $-$ & $[[14, 6, 2]]$ & $-$ & $-$ & $-$ & $-$ \\ \hline
			$[16,16,8]$ & $-$& $[[16, 6, 2]]$ & $-$ & $[[32, 10, 2]]$ & $-$ & $-$  \\ \hline
			$[18,18,6]$ & $-$ & $[[18, 6, 2]]$ & $-$ & $[[36, 10, 2]]$ & $-$ & $-$ \\ \hline
			$[18,18,18]$ & $-$ & $-$ & $-$ & $-$& $-$& $[[54,20,2]]$ \\ \hline
			$[20,20,10]$ & $-$& $-$ & $-$ & $-$ & $[[10, 5, 2]]$ & $[[60, 20, 2]]$ \\ \hline
			$[24,24,6]$ & $-$ & $-$ & $-$ & $-$ & $[[12, 5, 2]]$ & $[[72, 20, 2]]$ \\ \hline
			$[24,24,12]$ & $-$ & $-$ & $-$ & $[[24, 10, 2]]$ & $-$ & $-$ \\
				\end{tabular}
		\end{ruledtabular}
	\end{table*}

	\section{Final Remarks}\label{FR}
	
In this paper, we have constructed several new quantum Floquet codes. To do this task, we made use of the derivation tessellation and applied the hyperbolic regular tessellations to obtain semi-regular compact (orientable and non-orientable) surfaces of genus $g\geq 2$ in which the Floquet codes were constructed. Working with derived semi-regular tessellations, one gets an increase and a diversity of possibilities to construct new quantum Floquet codes. In Tables $2$ to $6$, we have presented several new Floquet codes constructed from orientable surfaces by means of these two techniques of tessellation derivation: clipping and incenter, for genus $1$ to $5$. One can observe that there is a uniform behavior regarding the coding rate in relation to tessellations. We also utilize similar techniques to construct codes over non-orientable surfaces. For instance, in Tables $8$ to $10$, we presented some new Floquet codes constructed from non-orientable surfaces by means of the same two techniques of tessellation derivation for genus $3$, $5$, and $7$.
	
The construction of codes over non-orientable surfaces can be viewed as a type of generalization compared to codes constructed over orientable surfaces in the sense that every code over an orientable surface provides an analogous code over a non-orientable surface. Furthermore, when comparing codes generated over these two types of surfaces, i.e., orientable and non-orientable, the codes over orientable surfaces have a better rate, $kd^2/n$, whereas the codes over non-orientable surfaces have a better coding rate, $k/n$.

	\end{document}